\documentclass[10pt,aps,prl,superscriptaddress,twocolumn,longbibliography,floatfix]{revtex4-2}
\usepackage{graphicx} %
\usepackage[colorlinks=true,linkcolor=blue, citecolor=blue]{hyperref}
\usepackage{amssymb}
\usepackage{braket}
\usepackage{algorithm2e}
\usepackage{stmaryrd}
\usepackage{orcidlink}
\usepackage{amsthm}
\usepackage{thmtools, thm-restate}
\usepackage{physics}
\usepackage{xcolor}
\usepackage{mathtools}
\usepackage{lipsum}

\usepackage{amsmath}

\usepackage[scr=rsfs,cal=boondox]{mathalfa}

\usepackage{multirow}
\usepackage{booktabs}

\SetKwRepeat{Do}{do}{while}
\newcommand{\poly}{\mathrm{poly}}

\newcommand{\dsl}[0]{\llbracket}
\newcommand{\dsr}[0]{\rrbracket}

\newcommand{\prlsection}[1]{{\em {#1}.---}}

\newcommand{\negl}{\operatorname{negl}}

\renewcommand{\eqref}[1]{Eq.~(\ref{#1})}
\usepackage{cleveref}

\usepackage{float}

\DeclareMathAlphabet{\mathcal}{OMS}{cmsy}{m}{n}

\definecolor{THc}{rgb}{0.9,0.3,0.2}

\newcommand{\Eset}[1]{\underset{#1}{\mathbb{E}}}
\newcommand{\SMLong}{Appendix}

 \definecolor{rahulcolor}{rgb}{0.5,0, 0.5}

\definecolor{nbcolor}{rgb}{0.2, 0.5, 0.5}

\definecolor{KB}{rgb}{0.4,0.3,0.9}

\definecolor{daxcolor}{rgb}{0.913,0.188,0.478}

\newtheorem{theorem}{Theorem}%

\newtheorem{lemma}{Lemma}
\newtheorem{definition}{Definition}
\newtheorem{proposition}{Proposition}

\newtheorem{corollary}{Corollary}

\floatname{algorithm}{Protocol}
\newcommand{\ihpc}{Institute of High Performance Computing (IHPC), Agency for Science, Technology and Research (A*STAR), 1 Fusionopolis Way, \#16-16 Connexis, Singapore 138632, Republic of Singapore}
\newcommand{\qinc}{Quantum Innovation Centre (Q.InC), Agency for Science, Technology and Research (A*STAR), 2 Fusionopolis Way, Innovis \#08-03, Singapore 138634, Republic of Singapore}
\newcommand{\sutd}{Science, Mathematics and Technology Cluster, Singapore University of Technology and Design, 8 Somapah Road, Singapore 487372, Republic of Singapore\looseness=-1}
\newcommand{\warwick}{Department of Computer Science, University of Warwick, Coventry, UK}

\begin{document}

\title{Quantum Error Correction in Adversarial Regimes}

\author{Rahul Arvind}
\email{rahul\_arvind@utexas.edu}
\thanks{equal contribution}
\affiliation{%
Department of Physics, University of Texas at Austin, Austin, TX 78712, USA
}%

\affiliation{\ihpc}
\author{Nikhil Bansal} 
\email{nikhilbansaliiser@gmail.com}
\thanks{equal contribution}
\affiliation{\ihpc}
\affiliation{\warwick}

\author{Dax Enshan Koh\,\orcidlink{0000-0002-8968-591X}
}
\email{dax\_koh@a-star.edu.sg}
\affiliation{\ihpc}
\affiliation{\qinc}
\affiliation{\sutd}
\author{Tobias Haug\,\orcidlink{0000-0003-2707-9962}}
\email{tobias.haug@u.nus.edu}
\affiliation{Quantum Research Center, Technology Innovation Institute, Abu Dhabi, UAE}
\author{Kishor Bharti\,\orcidlink{0000-0002-7776-6608}}
\email{kishor.bharti1@gmail.com}
\affiliation{\qinc}
\affiliation{\ihpc}

\begin{abstract}
In adversarial settings, where attackers can deliberately and strategically corrupt quantum data, standard quantum error correction reaches its limits. It can only correct up to half the code distance and must output a unique answer. Quantum list decoding offers a promising alternative. By allowing the decoder to output a short list of possible errors, it becomes possible to tolerate far more errors, even under worst-case noise. But two fundamental questions remain: which quantum codes support list decoding, and can we design decoding schemes that are secure against efficient, computationally bounded adversaries? In this work, we answer both. 
To identify which codes are list-decodable, we provide a generalized version of the Knill-Laflamme conditions. Then, using tools from quantum cryptography, we build an unambiguous list decoding protocol based on pseudorandom unitaries. Our scheme is secure against any quantum polynomial-time adversary, even across multiple decoding attempts, in contrast to previous schemes. Our approach connects coding theory with complexity-based quantum cryptography, paving the way for secure quantum information processing in adversarial settings.
\end{abstract}

\maketitle

 \let\oldaddcontentsline\addcontentsline%
\renewcommand{\addcontentsline}[3]{}%

As nations race to build quantum computers, the question of protecting them under adversarial conditions becomes pressing, especially in the context of cyber warfare. In this scenario, the adversary is not nature. It is a well-informed opponent who knows the protocol, anticipates defenses, and carefully injects errors to undermine both quantum communication and quantum computation. This is not science fiction. It is a practical threat. 

At the heart of this challenge lies a foundational tension in coding theory. Shannon's probabilistic noise model~\cite{shannon1948mathematical} assumes that errors are random and independent. Under this view, error correction relies on statistical regularities and high-probability behavior. In contrast, Hamming~\cite{hamming1950} considered worst-case scenarios where an adversary corrupts the transmitted information deliberately, subject only to a limit on the number of errors. These two perspectives define fundamentally different worlds. One is forgiving and stochastic; the other is adversarial and combative.

Building on this tension, list decoding offers a powerful bridge between the forgiving world of random noise and the hostile world of adversarial corruption~\cite{tal2015list, silverberg2003applications, morillo2009security, mihaljevi2001fast}. Introduced independently by Elias~\cite{elias1957list} and Wozencraft~\cite{wozencraft1958list} in the 1950s, list decoding relaxes the requirement of unique decoding. Instead of demanding a single valid output, the decoder returns a short list of candidates, one of which is guaranteed to be correct. This small compromise leads to a large gain. List decoding allows correction of nearly twice as many errors compared to unique decoding, breaking through the traditional half-distance barrier~\cite{guruswami2002combinatorial, guruswami2003list}. Even in the presence of worst-case noise patterns, reliable recovery becomes possible. In typical scenarios, the decoder still outputs a single predicted error, but in adversarial cases, the list ensures the correct answer is captured. This paradigm has reshaped classical coding theory, allowing rates and error thresholds once thought incompatible~\cite{guruswami2004list, peikert2006cryptographic}.

Quantum list decoding, introduced in early works such as those by Kawachi and Yamakami~\cite{kawachi_2006} and later refined in the context of adversarial quantum channels by Leung and Smith~\cite{leungsmith}, extends these ideas to the quantum domain~\cite{bergamaschi2022approachingquantumsingletonbound,bergamaschi2024listdecodablequantumldpc}. But fundamental questions remain. First, when is quantum list decoding possible? That is, which quantum codes can be list decoded, and under what conditions? This structural question has not yet been answered. Second, in practice, adversaries are not all-powerful. They are computationally bounded. Can we use complexity-theoretic tools to design quantum error correction schemes that remain secure against efficient, bounded adversaries? Both questions, structural and cryptographic, have remained open. In this work, we address both.

We begin by providing a general characterization of quantum list codes by providing the analogue of Knill-Laflamme conditions~\cite{PhysRevA.55.900} for quantum list codes. Next, we explore the adversarial setting from a cryptographic lens. We present an unambiguous list decoding scheme using pseudorandom unitaries~\cite{ji2018pseudorandom, ma2025constructrandomunitaries} that remains secure against any quantum polynomial-time adversary, even under multiple communication rounds.  %
Our techniques bridge coding theory and complexity-based quantum cryptography, opening a path toward secure quantum information processing in hostile environments.

\prlsection{Quantum list codes} 
Let $\mathcal{P}_n$ be the Pauli group on $n$ qubits, and let $\mathcal{G}(Q)$ be an abelian subgroup of $\mathcal{P}_n$ generated by Pauli operators $\langle g_1, g_2, \ldots, g_r \rangle$. Then, a $\dsl n, k\dsr$ stabilizer code $Q$ with $k$ logical qubits is stabilized by $\mathcal{G}(Q)$ with $r = n-k$ number of generators, with its code space given by $Q = \{|\psi\rangle \in \mathbb{C}^{\otimes n}_2: g_i|\psi\rangle = |\psi\rangle,  \forall i \in \{1, 2, \ldots, r\}\}$. Further, we define the weight of an error operator $E \in \mathcal{P}_n$ as the number of qubits it acts on nontrivially. 

Given any stabilizer code, the normalizer group $\mathcal{N}(Q)$ is defined as the subgroup of $\mathcal{P}_n$ that commutes with all the elements of $\mathcal{G}(Q)$. The undetectable errors of the stabilizer code are exactly the elements in $\mathcal{N}(Q) - \mathcal{G}(Q)$. Then we can define a $\dsl n, k, d \dsr$ code as an $\dsl n , k \dsr$ stabilizer code $Q$ with code distance $d$, given by the lowest weight of an operator in $\mathcal{N}(Q) - \mathcal{G}(Q)$. %

Finally, let us define the syndrome measurement. Let $E$ be an error on the code space. The syndrome associated with $E$ is the set of measurement outcomes for all the $g  \in \mathcal{G}(Q)$. Further, two operators are said to be stabilizer distinct if one cannot find a $g \in \mathcal{G}(Q)$ that transforms one into another. Given a measured syndrome, a decoder~\cite{dennis2002topological} finds the most likely error that has occurred, which is then corrected by applying the corresponding correction operation. If this can be done with high probability for a set of errors, then the code \textit{unambiguously decodes} that set.

An $\dsl n, k, d \dsr$-code can uniquely correct errors up to weight $(d - 1)/2$. 
By relaxing the condition of unique correction, one can extend the range of errors that can be corrected. In particular, one can define quantum list codes where for every error $E$, decoding finds a list of (up to) $L$ possible errors:
\begin{definition}[Quantum list codes~\cite{bergamaschi2022approachingquantumsingletonbound}]\label{def:QLD}
Let $Q$ be an $\dsl n,k,d \dsr$ quantum stabilizer code and let $\mathcal{E} \subset P_{n}$ be a set of errors. Then $Q$ is said to be $L$-QLD (quantum list decodable) if for every $E$ in $\mathcal{E}$ we have at most $L$ stabilizer distinct errors in $\mathcal{E}$ which agree with the syndrome of $E$. Furthermore, if $\mathcal{E}$ consists of all Pauli operators with weight at most $\epsilon \cdot n$, then we call the code to be $(\epsilon, L)$-list decodable.
\end{definition}
Essentially, quantum list codes correct up to a finite list of possible errors.
Let us say that we have a code space state $\ket{\psi} \in Q$. Then, if some adversary corrupts the state with some operator $E^{s}_{j}$, where $j \in [L]$ and $s$ indexes the syndrome, then there exist \textit{at most} $L$ stabilizer distinct operators that the adversary could have used. Consequently, the decoder will output the list errors that share the same syndrome as $E_j^s$. Since each error is equally likely, we essentially get a mixed state as the decoding output:
\begin{equation}
    \rho = c(s) \sum_{E_i^s \in E_j^s\mathcal{N}(Q) \cap \mathcal{E}} E_i^s |\bar \psi\rangle \langle \bar \psi| E_i^{s, \dagger}
\end{equation}
with $c(s)$ as normalizing factor dependent on syndrome $s$.
We emphasize that since $Q$ is a stabilizer quantum error correction code, it satisfies the usual Knill-Laflamme relations for some set of errors $\mathcal{E^{\prime}}$ which it can correct unambiguously (given by the distance of the stabilizer code). The implicit assumption here is that we allow a larger set of errors $\mathcal{E}$ than what $Q$ can correct unambiguously, while relaxing the requirement of unique decoding, so that list decoding is a generalization of the usual quantum error correction. %

While unambiguous decoding cannot correct more than $\frac{1-R}{4}$ symbols given a code rate $R$, there exist quantum list codes that can get to $\frac{1-R}{2}$~\cite{bergamaschi2022approachingquantumsingletonbound}. This is similar to the classical case discussed in~\cite{peikert2006cryptographic}, where it is shown that list error correction with Reed-Solomon codes is capable of error rates of up to $1-R$. These classical list codes can be converted to unambiguous error correction codes at the cost of extending the message length~\cite{peikert2006cryptographic} 
. One of the main results in this paper is that one can provide a somewhat analogous construction (Algorithm~\ref{alg:unambiguous_qld}) in the quantum case with adaptive security.

\prlsection{Pseudorandom Unitaries} Pseudorandom unitaries (PRUs) are elements of a keyed family of efficiently preparable unitaries that are computationally indistinguishable from the Haar random unitaries~\cite{ji2018pseudorandom}.
Here, computational indistinguishability means that there is no adversary with bounded computational power (i.e. quantum algorithms with polynomial depth and qubits) that can distinguish between the two ensembles. %
It turns out that such PRUs can be generated in $\text{polylog}(n)$ depth~\cite{schuster2024randomunitariesextremelylow, ma2025constructrandomunitaries} in contrast to the truly random (Haar measure) unitaries, which require exponential depth. A precise definition of PRUs is provided in Def.~\ref{def:PRU}.

\prlsection{Knill-Laflamme condition for quantum list codes}
\label{knill-laflamme}
In this section, we generalize the well-known Knill-Laflamme conditions to quantum list codes. These provide a necessary and sufficient condition for the quantum error correction codes to be able to correct a set of errors: 

\begin{theorem}[Knill-Laflamme condition for list decoding]\label{eq:knilllaflammList}
   Let $Q$ be a $L$-$\mathrm{QLD}$ code for error set $\mathcal{E}= \{E_{i}\}_i$, and let \texttt{Syn} be the associated syndrome operator. Let $\Pi_Q$ be the projector onto the code space. Then, when two errors $E$ and $F$ in $\mathcal{E}$ have different syndromes, we have
\begin{equation}
    \Pi_QE^{\dag}F\Pi_Q = 0\,.
\end{equation}
In contrast, when $E$ and $F$ have the same syndrome $s$, then we have
\begin{equation}
    \Pi_Q E^{\dag} F\Pi_Q = c(s) \alpha(E,F)\Pi_Q,
\end{equation}
where $\alpha(E,F) = \sum_{k,i,j} U^{*}(E)_{ki}\Pi_Q E_{i}^{\dag}U(F)_{kj}E_{j}$ depend on $E,F$ with $U(E)$ and $U(F)$ being unitary matrices and $c(s) = 1/l(s)$ is the inverse of number of list errors corresponding to syndrome $s$.

\end{theorem}

We prove the Knill-Laflamme conditions in Appendix~\ref{appendix_KL_conditions_proof}.

\prlsection{Unambiguous cryptographic decoding of list codes}
\label{list_decoding_unambiguous}
Now, we provide a method to implement unambiguous decoding of quantum list codes using pseudorandom unitaries. %

\prlsection{Protocol} Let $\ket{{\phi}}$ %
be an $n$-qubit logical state we wish to protect, and further let $Q_L$ be a $\dsl n^{\prime},n+m,d \dsr$ stabilizer code such that $Q_L$ is $L-\mathrm{QLD}$ for error channel $\mathcal{E}$. The protocol starts by appending $m$-qubit tag-state $\ket{0}^{\otimes m}$ to  $\ket{{\phi}} $, and then applying an $(n+m)$-qubit pseudorandom unitary $U_{\lambda}$ %
on the combined state, giving us $\ket{{\psi}}=U_{\lambda}\ket{\phi}\ket{0}^{\otimes m}$. %
We then encode this state as a logical state into the $L$-QLD code $Q_L$ with $n'$ physical qubits.

The adversary can apply errors from the set $\mathcal{E}$.
We will use $V_{Q_L}$ to represent the encoding map associated with $Q_L$, $\mathcal{D}_{Q_L}$ for the decoding map associated with $Q_L$, and encode $\ket{{\psi}}$ into $\ket{\bar\psi}$ using $V_{Q_L}$.

Now, let us assume that the adversary corrupts the state with an error which results in the measured syndrome $s$. The $L$-QLD then outputs a quantum state $ E_{j}^{s} \ket{\bar{\psi}}\bra{\bar{\psi}} E_{j}^{s \dag}$ with the guarantee that $E_j^s$ lies in the list of errors decided by the syndrome $s$ of size at most $L$. Unambiguous decoding then comes down to finding the correct list error. In Lemma~\ref{lemma:equiv_defn}, we show that for unambiguous decoding purposes, we can work entirely in the logical space with list errors on the logical space given by the corresponding list errors on the physical space. For simplicity, we denote the list errors acting on the logical space using the same labels $\{E_j^s\}_{j}$ as we used for the corresponding list errors in physical space.

Without any loss of generality, assume a specific error $E_{j}^{s}$ with index $j$ has occurred. From the list decoding, we only know that any error from the list $\{E_b^{s}\}_{b=1}^\ell$ with $\ell \leq L$ may have occurred. Our algorithm involves an exhaustive search over the list elements by applying the inverse of a chosen element of the list and then measuring the auxiliary system, until the correct list error is found.

It is easy to see that if we pick the correct error in first round, we get, $U_{\lambda}^\dagger E_j^{s \dag}E_{j}^{s} \ket{\psi}=\ket{\phi}\ket{0}^{\otimes m}$, which can be verified by using the two-outcome measurement $\{\Pi, I- \Pi\}$ with $\Pi = I_n \otimes |0\rangle\langle 0|^{\otimes m}$ on the auxiliary system. This allows us to reconstruct the initial state on the first register.

The ostensibly more complicated case arises when we select one of the wrong correction operators. After applying the inverse of the list element chosen $E_b^{s, \dagger}$ with $b\neq j$ as well as $U^{\dagger}_{\lambda}$, the resulting state $U_{\lambda} E_b^{s, \dagger} E_{j}^{s} \ket{\psi}$ has low fidelity with $\ket{\phi}\ket{0}^{\otimes m}$. Then, the two-outcome measurement is applied, where $\Pi=I_n \otimes \ketbra{0}^{\otimes m}$ fails with high probability. This is due to the fact that $U_\lambda$ is a $2$-design, which results in the $m$ auxiliary qubits being highly entangled with the $n$ qubits and far from the $\ket{0}^{\otimes m}$ state. Instead,  with high probability, the projection $I-\Pi$ is applied. The projected state is only weakly perturbed, and the original state $E_j^s \ket{\psi}$ can be recovered with high fidelity by applying  $U_\lambda$ and the list error $E_b^{s}$ to recover the original state with high fidelity. Then, we can proceed to select the next list error $E_{b^{\prime}}^s$, repeating until we find the correct list element $E_j^s$. %

\RestyleAlgo{boxruled}
\LinesNumbered
\IncMargin{0.8em}
\begin{algorithm}
    \label{alg:unambiguous_qld}
  \caption{Unambiguous quantum list decoding}
  \KwData{Input $n^{\prime}-$qubit corrupted state $\rho$, $L-$QLD code $\dsl n', n+m, d\dsr$, key $\lambda$, set of possible errors $\mathcal{E}$}
  \KwResult{Output $\rho\approx\ket{\phi}\bra{\phi}\otimes(\ketbra{0}{0})^{\otimes m}$}
  Measure syndrome $s$ and list decode and denote possible (logical) errors $\mathcal{E}^{s}=\{E_b^{s}\}_{b=1}^\ell$ where $\ell\leq L$\\ 
  \For{$b=1,\dots,\vert \mathcal{E}^{s}\vert$}{
    Apply $U_{\lambda}^\dagger {E_b^{s}}^{\dagger}$%
    and use the two-outcome projection $\{\Pi, I-\Pi\}$ on the state.
    \\
    \lIf {projection succeeds}{
    
    \hspace{0.5cm}\Return post measurement state
    }
    \lElse{
    
     \hspace{0.5cm} Apply ${E_b^{s}} U_{\lambda}$ on the state
    }}
\end{algorithm}

\prlsection{Correctness of the protocol}
In this section, we show that our unambiguous decoding algorithm finds the correct list error. %
As mentioned previously, in the case that the receiver chooses the correct list error, it directly obtains the correct logical state. 
In contrast, when the receiver has selected one of the wrong list errors, the following must hold:
\begin{enumerate}
    \item The projection $\Pi$ must succeed with a very small probability that ideally decays with the system size.
    \item The receiver should be able to recover the initial state (affected by error $E_j^s$) with high fidelity after applying the inverse of the wrong correction operator $E_b^s$.
\end{enumerate}
The former condition ensures that the algorithm gives false positives with low probability, while the latter condition makes it possible to consecutively search through all errors within the list, thus making it possible to uniquely decode from a single copy of the list-error state.

Let us start with the first of these conditions. We will use the notation $\rho = \ket{\phi}\bra{\phi}$, and denote the state with PRU and error as $\rho_{\lambda, 0} = E_j^sU_{\lambda} \rho\otimes (\ketbra{0}{0})^{\otimes m} U_{\lambda}^{\dagger} E_j^{s}$. Assume that we choose the list element $E^{s, \dagger}_{1}$ with $j\neq 1$. After applying the inverse  PRU, $U_{\lambda}^{\dagger}$, we have the ``wrong" state 

\begin{equation}\label{eq:wrong_error_state_1}
    \sigma_{\lambda, 1} = U_{\lambda}^\dagger \bigg(E_1^{s \dag} \rho_{\lambda, 0} E_1^{s}\bigg)U_{\lambda}.
\end{equation}
Now, using the projection measurement $\{\Pi, I-\Pi\}$ on $\sigma_{\lambda, 1}$ with $\Pi$ as defined previously, the probability of measuring the outcome $I - \Pi$ is
\begin{equation}\label{eq:prob_overlap_with_0}
        P_{\lambda, 1}=1 -
        \mathrm{Tr} \bigg[\Pi U_{\lambda}^\dagger  \bigg(E^{s}_1 \rho_{\lambda, 0} E_1^{s}\bigg)  U_{\lambda} \bigg],
\end{equation}
where we drop the superscript $\dagger$ on the $E_i^s$ since they are Pauli errors. 
We now assume that PRU, $\{U_\lambda\}_{\lambda}$ is also an exact $2$-design, which can be implemented in $\text{polylog}(n)$ depth applying a $2$-design to a PRU~\cite{cleve2015near,haug2025pseudorandomquantumauthentication}. Then,  the average of $P_{\lambda, 1}$ over keys $\lambda$ can be calculated by the Haar average
\begin{equation}\label{eq:haar_prob}
    \mathbb{E}_\lambda P_{\lambda, 1} = \Eset{U} P_1(U) \,.
\end{equation}
Using this, we show that for ancilla size $m = \omega(\log n)$, the probability $P_{\lambda, 1}$ of observing outcome $I-\Pi$ on the state $\sigma_{\lambda, 1}$ as defined above in~\eqref{eq:wrong_error_state_1} is very high with overwhelming probability over the keys. Formally,
\begin{proposition}\label{thm:prob_calcs_main_text}
    Let $P_{\lambda, 1}$ be as defined in~\eqref{eq:prob_overlap_with_0}, then
    \begin{equation}
        P_{\lambda, 1} \geq 1-\delta \sim 1-\negl(n)
    \end{equation}
    with probability $1 - {O}(2^{-m}/\delta) \sim 1 - \negl(n)$ over keys for $\delta = \Theta(2^{-m+\poly\log n})$ and $m = \omega(\log n)$,
    where we assume that PRU, $\{U_\lambda\}_{\lambda}$ also forms  an exact $2$-design. 
\end{proposition}
\begin{proof}
    The proof follows from~\eqref{eq:haar_prob} by equating $\mathbb{E}_\lambda P_{\lambda, 1}$ using the corresponding Haar average and then using Markov's inequality.
    We provide the details of the proof in Appendix~\ref{appendix_wco}.
\end{proof}

Now, denote the state
\begin{equation}
    \rho_{\lambda, 1} = \frac{1}{P_{\lambda, 1}}E^{s}_1 U_{\lambda} (I-\Pi) U_{\lambda}^{\dagger}E^{s}_1 \rho_{\lambda, 0} E^{s}_1 U_{\lambda} (I-\Pi) U_{\lambda}^{\dagger} E^{s}_1,
\end{equation}
which is the state after applying the PRU and inverse of the list error chosen earlier on the post-measurement state selected on outcome $(I-\Pi)$. We find that $(I-\Pi)$ weakly perturbs $\rho_{\lambda, 1}$ compared to $\rho_{\lambda, 0}$, which we quantify with the Uhlmann fidelity $\mathcal{F}(\rho_{\lambda, 1}, \rho_{\lambda, 0}) = (\tr \sqrt{\sqrt{\rho_{\lambda, 0}}\rho_{\lambda, 1}\sqrt{\rho_{\lambda, 0}}})^2$ which can be simply written as $\mathcal{F}=\tr(\rho_{\lambda, 1} \rho_{\lambda, 0})$ since $\rho_{\lambda, 0}$ is pure, giving us %
\begin{multline}\label{eq:fidelity_reusability}
 \mathcal{F}(\rho_{\lambda, 1}, \rho_{\lambda, 0}) =  \bigg(\frac{1}{P_{\lambda, 1}}\Tr  \Big[\rho_{\lambda, 0}E_1^{s}   U_{\lambda}(I-\Pi)U_{\lambda}^{\dagger} \\ E^{s}_1 \rho_{\lambda, 0} E_1^{s}  U_{\lambda} (I-\Pi) U_{\lambda}^{\dagger}  E^{s}_1\Big] \bigg).
\end{multline}

We can lower bound the fidelity in \eqref{eq:fidelity_reusability} by using the Haar integration, with the additional assumption that PRU forms an approximate 4-design with relative error $\epsilon = 2^{-m}$. Formally, %
\begin{proposition}\label{prop:fidelity_calcs_main_text}
    The fidelity of the logical state after checking whether error $E_1^s$ from the list has occurred,  as given by~\eqref{eq:fidelity_reusability} is bounded by
    \begin{equation}
        \mathcal{F}(\rho_{\lambda, 1}, \rho_{\lambda, 0}) \geq 1 - \delta \sim 1 - \negl(n)
    \end{equation}
     with probability $1 - O\left(\frac{2^{-m} + \epsilon}{\delta}\right) \sim 1 - \negl(n)$ over keys assuming that $\{U_{\lambda}\}_{\lambda}$ forms an approximate 4-design with relative error $\epsilon = 2^{-m}$ and an exact $2$-design, for $\delta = \Theta(2^{-m + \text{polylog}(n)})$ and $m = \omega(\log n)$.
\end{proposition}
\begin{proof}
 The proof idea is the same as for $P_{\lambda, 1}$, with the additional assumption that $\{U_{\lambda}\}_{\lambda}$ is an approximate $4$-design. This allows one to write the fidelity in terms of the Haar average, followed by the use of Markov's inequality. We provide the full proof in Appendix~\ref{appendix_fidelity_calcs}. %
\end{proof}

Now, combining Proposition~\ref{thm:prob_calcs_main_text} and Proposition~\ref{prop:fidelity_calcs_main_text}, we present the analysis of the algorithm~\ref{alg:unambiguous_qld} in the following theorem. Informally, we show that for ancilla size $m = \omega(\log n)$, and assuming that PRU is an exact 2-design and an $\epsilon = O(2^{-m})$ relative-error approximate $4$-design, algorithm~\ref{alg:unambiguous_qld} returns a state which is very close to the original logical state with overwhelming probability over both the keys and randomness in the algorithm due to measurements. %

\begin{theorem}[Unambiguous quantum list decoding]\label{thm:almost_pure_state}
Let $Q_L$ be an $\dsl n^{\prime}, n+m\dsr$ stabilizer code which is also an $L$-list code for an error set $\mathcal{E}$. Then Algorithm~\ref{alg:unambiguous_qld} returns the correct logical state with fidelity at least $1 - \Theta(L 2^{-m/2 + \mathrm{polylog}(n)}) \sim 1-\negl(n)$ with success probability $1 - \negl(n)$ over the randomness in the algorithm and with probability $1 - O\left(\frac{\epsilon + 2^{-m}}{\delta}\right) \sim 1 - \negl(n)$ over keys, for $\delta = \Theta(2^{-m+\mathrm{polylog} (n)})$, $m = \omega(\log n)$ and $L = {O}(\poly(n))$
assuming that $U_{\lambda}$ is both an approximate $4$-design with relative error $\epsilon = O(2^{-m})$ and an exact $2$-design.
\end{theorem}

\begin{proof}
 The full proof is in Appendix~\ref{appendix_unique_decoding_proof}, while we sketch the main steps below:
 \begin{enumerate}
     \item The algorithm iterates through the list of possible errors, and identifies the correct error with probability close to 1, with high probability over the keys $\lambda$ and randomness in the measurement. %
     \item The state returned by the algorithm and the logical state have a fidelity which is lower bounded by $1 - \negl(n)$, with high probability over the keys $\lambda$ and randomness in the measurement.
 \end{enumerate}
\end{proof}

\prlsection{Cryptographic approximate quantum error correction from list-decodable codes}
Now, using the above theorem, we can construct error correction codes with strong security guarantees against (efficient) adversaries. In particular, we give a construction of cryptographic approximate quantum error correction (Definition~\ref{def:private_AQECC}) based on PRU:
\begin{corollary}[Cryptographic approximate quantum error correction from list-decodable codes]\label{corr:AQECC}%
    Let $Q_L$ be a $[[n^{\prime}, n+m]]$ stabilizer code which is also $(\alpha, L)$ list decodable with encoding and decoding operations as $\mathrm{Enc}_S$ and $\mathrm{Dec}_S$ respectively,  and let $\{U_{\lambda}\}_{\lambda}$ be PRU with exact $2$-design and relative error $\epsilon$ approximate $4$-design property. Then, define the pair of algorithms $\mathrm{Enc}$ and $\mathrm{Dec}$ as $\mathrm{Enc}(\lambda, |\phi\rangle \langle\phi|) = \mathrm{Enc}_S (U_{\lambda} |\phi\rangle \langle\phi| \otimes |0\rangle\langle0|^{\otimes m} U_{\lambda}^{\dagger})$ and decoding operation as $\mathrm{Dec}(\lambda, \rho)  = \mathrm{Dec}_A(\lambda, \mathrm{Dec}_{S}(\rho))$ where $\mathrm{Dec}_{ A}$ is the decoder as in the Algorithm~\ref{alg:unambiguous_qld}, then $(\mathrm{Enc}, \mathrm{Dec})$ is a $(\mathcal{A}, \sqrt{L\sqrt{2\delta}})$ cryptographic approximate quantum error correcting code where $\mathcal{A}$ is any $(\alpha,t)-n'$ qubit adversary as defined in Def.~\ref{def:advnoise} for $t = O(\poly (n))$, $\delta = \Theta(2^{-m + \poly \log n})$ is the fidelity error as in Proposition~\ref{prop:fidelity_calcs_main_text} and $\epsilon = O(2^{-m})$.
\end{corollary}
\begin{proof}
    The proof follows from using Theorem~\ref{thm:almost_pure_state}, the Fuchs–van de Graaf inequality between trace distance and fidelity, and properties of PRU to ensure security against adversaries. For completeness, we present the proof in Appendix~\ref{appendix_unique_decoding_proof}. 
\end{proof}
Note that our construction can be implemented in low circuit depth and negligible fidelity error, due to PRUs with $2$- and $4$-design properties requiring only $\mathrm{polylog}(n)$ depth~\cite{schuster2024randomunitariesextremelylow,haug2025pseudorandomquantumauthentication}.
We highlight that our code is secure against adversarial noise of a computationally bounded adversary, e.g as defined in Def.~\ref{def:advnoise}. This holds even when the same key $\lambda$ is reused many times, in contrast to previous works~\cite{leungsmith, bergamaschi2022approachingquantumsingletonbound}.

\prlsection{Discussion} Our work brings two powerful ideas together: quantum list decoding and computational cryptographic security. 
Despite significant progress in quantum list decoding, a central structural question remains unanswered: under what conditions is quantum list decoding possible? Our work addresses this gap by introducing general conditions for quantum list decodability, extending the Knill-Laflamme framework. 
Further, we introduce cryptographically secure construction of unique decoding for quantum list codes. Notably, our code can correct errors even when the noise is controlled by an adversary with computationally bounded power. Notably, this holds even when the key is re-used over many rounds of communications or the adversary learns from $\text{poly}(n)$ copies of the code, in contrast to previous constructions~\cite{leungsmith, bergamaschi2022approachingquantumsingletonbound}. 
We also note that our construction establishes a  purity testing code (Def.~\ref{def:ptc})~\cite{barnum2002authentication, portmann2017, dulek2018quantum} where the key can be re-used many times even when attacked by computationally efficient adversaries, which is  inspired by recently introduced pseudorandom quantum authentication schemes~\cite{haug2025pseudorandomquantumauthentication}. 
Altogether, our work offers an efficient approach to secure quantum coding, blending ideas from quantum error correction and complexity based cryptography.

Our work opens the door to many promising directions. While we focused on list decoding for stabilizer codes, similar list decoding algorithms could be studied for subsystem codes~\cite{aly2006subsystem, yoder2019optimal, klappenecker2007subsystem, bacon2006quantum}, Floquet codes~\cite{davydova2023floquet, tanggara2024simple, fahimniya2023fault, zhang2023x} or spacetime codes~\cite{bacon2017sparse,gottesman2022opportunities,delfosse2023spacetime}. %
Another exciting direction is list decoding for strategic codes, where the code can adapt to the adversary’s actions~\cite{tanggara2024strategic}. Moreover, the decoding thresholds for quantum list codes under various noise models could be studied numerically towards enabling practical implementations. Such studies could also inform the design of fault-tolerant quantum architectures that integrate cryptographic security guarantees, pushing us closer to realizing secure and reliable quantum technologies.

\label{outlook}

\prlsection{Acknowledgments} NB thanks Matthias Caro for helpful discussion. KB and DEK acknowledge support from A*STAR C230917003 and the Q.InC Strategic Research and Translational Thrust. 
DEK acknowledges funding support from A*STAR under the Central Research Fund (CRF) Award for Use-Inspired Basic Research (UIBR).

\bibliography{references}

\begin{thebibliography}{43}%
\makeatletter
\providecommand \@ifxundefined [1]{%
 \@ifx{#1\undefined}
}%
\providecommand \@ifnum [1]{%
 \ifnum #1\expandafter \@firstoftwo
 \else \expandafter \@secondoftwo
 \fi
}%
\providecommand \@ifx [1]{%
 \ifx #1\expandafter \@firstoftwo
 \else \expandafter \@secondoftwo
 \fi
}%
\providecommand \natexlab [1]{#1}%
\providecommand \enquote  [1]{``#1''}%
\providecommand \bibnamefont  [1]{#1}%
\providecommand \bibfnamefont [1]{#1}%
\providecommand \citenamefont [1]{#1}%
\providecommand \href@noop [0]{\@secondoftwo}%
\providecommand \href [0]{\begingroup \@sanitize@url \@href}%
\providecommand \@href[1]{\@@startlink{#1}\@@href}%
\providecommand \@@href[1]{\endgroup#1\@@endlink}%
\providecommand \@sanitize@url [0]{\catcode `\\12\catcode `\$12\catcode `\&12\catcode `\#12\catcode `\^12\catcode `\_12\catcode `\%12\relax}%
\providecommand \@@startlink[1]{}%
\providecommand \@@endlink[0]{}%
\providecommand \url  [0]{\begingroup\@sanitize@url \@url }%
\providecommand \@url [1]{\endgroup\@href {#1}{\urlprefix }}%
\providecommand \urlprefix  [0]{URL }%
\providecommand \Eprint [0]{\href }%
\providecommand \doibase [0]{https://doi.org/}%
\providecommand \selectlanguage [0]{\@gobble}%
\providecommand \bibinfo  [0]{\@secondoftwo}%
\providecommand \bibfield  [0]{\@secondoftwo}%
\providecommand \translation [1]{[#1]}%
\providecommand \BibitemOpen [0]{}%
\providecommand \bibitemStop [0]{}%
\providecommand \bibitemNoStop [0]{.\EOS\space}%
\providecommand \EOS [0]{\spacefactor3000\relax}%
\providecommand \BibitemShut  [1]{\csname bibitem#1\endcsname}%
\let\auto@bib@innerbib\@empty
\bibitem [{\citenamefont {Shannon}(1948)}]{shannon1948mathematical}%
  \BibitemOpen
  \bibfield  {author} {\bibinfo {author} {\bibfnamefont {C.~E.}\ \bibnamefont {Shannon}},\ }\bibfield  {title} {\bibinfo {title} {A mathematical theory of communication},\ }\href {https://doi.org/10.1002/j.1538-7305.1948.tb01338.x} {\bibfield  {journal} {\bibinfo  {journal} {The Bell System Technical Journal}\ }\textbf {\bibinfo {volume} {27}},\ \bibinfo {pages} {379} (\bibinfo {year} {1948})}\BibitemShut {NoStop}%
\bibitem [{\citenamefont {Hamming}(1950)}]{hamming1950}%
  \BibitemOpen
  \bibfield  {author} {\bibinfo {author} {\bibfnamefont {R.~W.}\ \bibnamefont {Hamming}},\ }\bibfield  {title} {\bibinfo {title} {Error detecting and error correcting codes},\ }\href {https://doi.org/10.1002/j.1538-7305.1950.tb00463.x} {\bibfield  {journal} {\bibinfo  {journal} {The Bell System Technical Journal}\ }\textbf {\bibinfo {volume} {29}},\ \bibinfo {pages} {147} (\bibinfo {year} {1950})}\BibitemShut {NoStop}%
\bibitem [{\citenamefont {Tal}\ and\ \citenamefont {Vardy}(2015)}]{tal2015list}%
  \BibitemOpen
  \bibfield  {author} {\bibinfo {author} {\bibfnamefont {I.}~\bibnamefont {Tal}}\ and\ \bibinfo {author} {\bibfnamefont {A.}~\bibnamefont {Vardy}},\ }\bibfield  {title} {\bibinfo {title} {List decoding of polar codes},\ }\href {https://doi.org/10.1109/TIT.2015.2410251} {\bibfield  {journal} {\bibinfo  {journal} {IEEE Transactions on Information Theory}\ }\textbf {\bibinfo {volume} {61}},\ \bibinfo {pages} {2213} (\bibinfo {year} {2015})}\BibitemShut {NoStop}%
\bibitem [{\citenamefont {Silverberg}\ \emph {et~al.}(2003)\citenamefont {Silverberg}, \citenamefont {Staddon},\ and\ \citenamefont {Walker}}]{silverberg2003applications}%
  \BibitemOpen
  \bibfield  {author} {\bibinfo {author} {\bibfnamefont {A.}~\bibnamefont {Silverberg}}, \bibinfo {author} {\bibfnamefont {J.}~\bibnamefont {Staddon}},\ and\ \bibinfo {author} {\bibfnamefont {J.}~\bibnamefont {Walker}},\ }\bibfield  {title} {\bibinfo {title} {Applications of list decoding to tracing traitors},\ }\href {https://doi.org/10.1109/TIT.2003.810630} {\bibfield  {journal} {\bibinfo  {journal} {IEEE Transactions on Information Theory}\ }\textbf {\bibinfo {volume} {49}},\ \bibinfo {pages} {1312} (\bibinfo {year} {2003})}\BibitemShut {NoStop}%
\bibitem [{\citenamefont {Morillo}\ and\ \citenamefont {R{\`a}fols}(2009)}]{morillo2009security}%
  \BibitemOpen
  \bibfield  {author} {\bibinfo {author} {\bibfnamefont {P.}~\bibnamefont {Morillo}}\ and\ \bibinfo {author} {\bibfnamefont {C.}~\bibnamefont {R{\`a}fols}},\ }\bibfield  {title} {\bibinfo {title} {The security of all bits using list decoding},\ }in\ \href {https://doi.org/10.1007/978-3-642-00468-1_2} {\emph {\bibinfo {booktitle} {International Workshop on Public Key Cryptography}}}\ (\bibinfo {organization} {Springer},\ \bibinfo {year} {2009})\ pp.\ \bibinfo {pages} {15--33}\BibitemShut {NoStop}%
\bibitem [{\citenamefont {Mihaljevi}\ \emph {et~al.}(2002)\citenamefont {Mihaljevi}, \citenamefont {Fossorier},\ and\ \citenamefont {Imai}}]{mihaljevi2001fast}%
  \BibitemOpen
  \bibfield  {author} {\bibinfo {author} {\bibfnamefont {M.~J.}\ \bibnamefont {Mihaljevi}}, \bibinfo {author} {\bibfnamefont {M.~P.~C.}\ \bibnamefont {Fossorier}},\ and\ \bibinfo {author} {\bibfnamefont {H.}~\bibnamefont {Imai}},\ }\bibfield  {title} {\bibinfo {title} {Fast correlation attack algorithm with list decoding and an application},\ }in\ \href {https://doi.org/10.1007/3-540-45473-X_17} {\emph {\bibinfo {booktitle} {Fast Software Encryption}}},\ \bibinfo {editor} {edited by\ \bibinfo {editor} {\bibfnamefont {M.}~\bibnamefont {Matsui}}}\ (\bibinfo  {publisher} {Springer Berlin Heidelberg},\ \bibinfo {address} {Berlin, Heidelberg},\ \bibinfo {year} {2002})\ pp.\ \bibinfo {pages} {196--210}\BibitemShut {NoStop}%
\bibitem [{\citenamefont {Elias}(1957)}]{elias1957list}%
  \BibitemOpen
  \bibfield  {author} {\bibinfo {author} {\bibfnamefont {P.}~\bibnamefont {Elias}},\ }\bibfield  {title} {\bibinfo {title} {List decoding for noisy channels},\ }\href {http://hdl.handle.net/1721.1/4484} {\bibfield  {journal} {\bibinfo  {journal} {Technical report (Massachusetts Institute of Technology. Research Laboratory of Electronics)}\ }\textbf {\bibinfo {volume} {335}},\ \bibinfo {pages} {94} (\bibinfo {year} {1957})}\BibitemShut {NoStop}%
\bibitem [{\citenamefont {Wozencraft}(1958)}]{wozencraft1958list}%
  \BibitemOpen
  \bibfield  {author} {\bibinfo {author} {\bibfnamefont {J.~M.}\ \bibnamefont {Wozencraft}},\ }\bibfield  {title} {\bibinfo {title} {List decoding},\ }\href@noop {} {\bibfield  {journal} {\bibinfo  {journal} {Quarterly Progress Report}\ }\textbf {\bibinfo {volume} {48}},\ \bibinfo {pages} {90} (\bibinfo {year} {1958})}\BibitemShut {NoStop}%
\bibitem [{\citenamefont {Guruswami}\ \emph {et~al.}(2002)\citenamefont {Guruswami}, \citenamefont {Hastad}, \citenamefont {Sudan},\ and\ \citenamefont {Zuckerman}}]{guruswami2002combinatorial}%
  \BibitemOpen
  \bibfield  {author} {\bibinfo {author} {\bibfnamefont {V.}~\bibnamefont {Guruswami}}, \bibinfo {author} {\bibfnamefont {J.}~\bibnamefont {Hastad}}, \bibinfo {author} {\bibfnamefont {M.}~\bibnamefont {Sudan}},\ and\ \bibinfo {author} {\bibfnamefont {D.}~\bibnamefont {Zuckerman}},\ }\bibfield  {title} {\bibinfo {title} {Combinatorial bounds for list decoding},\ }\href {https://doi.org/10.1109/18.995539} {\bibfield  {journal} {\bibinfo  {journal} {IEEE Transactions on Information Theory}\ }\textbf {\bibinfo {volume} {48}},\ \bibinfo {pages} {1021} (\bibinfo {year} {2002})}\BibitemShut {NoStop}%
\bibitem [{\citenamefont {Guruswami}(2003)}]{guruswami2003list}%
  \BibitemOpen
  \bibfield  {author} {\bibinfo {author} {\bibfnamefont {V.}~\bibnamefont {Guruswami}},\ }\bibfield  {title} {\bibinfo {title} {List decoding with side information},\ }in\ \href {https://doi.org/10.1109/CCC.2003.1214429} {\emph {\bibinfo {booktitle} {18th IEEE Annual Conference on Computational Complexity, 2003. Proceedings.}}}\ (\bibinfo {organization} {IEEE},\ \bibinfo {year} {2003})\ pp.\ \bibinfo {pages} {300--309}\BibitemShut {NoStop}%
\bibitem [{\citenamefont {Guruswami}(2004)}]{guruswami2004list}%
  \BibitemOpen
  \bibfield  {author} {\bibinfo {author} {\bibfnamefont {V.}~\bibnamefont {Guruswami}},\ }\href {https://doi.org/10.1007/b104335} {\emph {\bibinfo {title} {List decoding of error-correcting codes: winning thesis of the 2002 ACM doctoral dissertation competition}}},\ Vol.\ \bibinfo {volume} {3282}\ (\bibinfo  {publisher} {Springer Science \& Business Media},\ \bibinfo {year} {2004})\BibitemShut {NoStop}%
\bibitem [{\citenamefont {Peikert}(2006)}]{peikert2006cryptographic}%
  \BibitemOpen
  \bibfield  {author} {\bibinfo {author} {\bibfnamefont {C.~J.}\ \bibnamefont {Peikert}},\ }\emph {\bibinfo {title} {Cryptographic error correction}},\ \href {http://hdl.handle.net/1721.1/38320} {Ph.D. thesis},\ \bibinfo  {school} {Massachusetts Institute of Technology} (\bibinfo {year} {2006})\BibitemShut {NoStop}%
\bibitem [{\citenamefont {Kawachi}\ and\ \citenamefont {Yamakami}(2006)}]{kawachi_2006}%
  \BibitemOpen
  \bibfield  {author} {\bibinfo {author} {\bibfnamefont {A.}~\bibnamefont {Kawachi}}\ and\ \bibinfo {author} {\bibfnamefont {T.}~\bibnamefont {Yamakami}},\ }\bibfield  {title} {\bibinfo {title} {Quantum hardcore functions by complexity-theoretical quantum list decoding},\ }in\ \href {https://doi.org/10.1007/11787006_19} {\emph {\bibinfo {booktitle} {Automata, Languages and Programming}}},\ \bibinfo {editor} {edited by\ \bibinfo {editor} {\bibfnamefont {M.}~\bibnamefont {Bugliesi}}, \bibinfo {editor} {\bibfnamefont {B.}~\bibnamefont {Preneel}}, \bibinfo {editor} {\bibfnamefont {V.}~\bibnamefont {Sassone}},\ and\ \bibinfo {editor} {\bibfnamefont {I.}~\bibnamefont {Wegener}}}\ (\bibinfo  {publisher} {Springer Berlin Heidelberg},\ \bibinfo {address} {Berlin, Heidelberg},\ \bibinfo {year} {2006})\ pp.\ \bibinfo {pages} {216--227}\BibitemShut {NoStop}%
\bibitem [{\citenamefont {Leung}\ and\ \citenamefont {Smith}(2008)}]{leungsmith}%
  \BibitemOpen
  \bibfield  {author} {\bibinfo {author} {\bibfnamefont {D.}~\bibnamefont {Leung}}\ and\ \bibinfo {author} {\bibfnamefont {G.}~\bibnamefont {Smith}},\ }\bibfield  {title} {\bibinfo {title} {Communicating over adversarial quantum channels using quantum list codes},\ }\href {https://doi.org/10.1109/TIT.2007.913433} {\bibfield  {journal} {\bibinfo  {journal} {IEEE Transactions on Information Theory}\ }\textbf {\bibinfo {volume} {54}},\ \bibinfo {pages} {883} (\bibinfo {year} {2008})}\BibitemShut {NoStop}%
\bibitem [{\citenamefont {Bergamaschi}\ \emph {et~al.}(2024{\natexlab{a}})\citenamefont {Bergamaschi}, \citenamefont {Golowich},\ and\ \citenamefont {Gunn}}]{bergamaschi2022approachingquantumsingletonbound}%
  \BibitemOpen
  \bibfield  {author} {\bibinfo {author} {\bibfnamefont {T.}~\bibnamefont {Bergamaschi}}, \bibinfo {author} {\bibfnamefont {L.}~\bibnamefont {Golowich}},\ and\ \bibinfo {author} {\bibfnamefont {S.}~\bibnamefont {Gunn}},\ }\bibfield  {title} {\bibinfo {title} {Approaching the quantum singleton bound with approximate error correction},\ }in\ \href {https://doi.org/10.1145/3618260.3649680} {\emph {\bibinfo {booktitle} {Proceedings of the 56th Annual ACM Symposium on Theory of Computing}}},\ \bibinfo {series and number} {STOC 2024}\ (\bibinfo  {publisher} {Association for Computing Machinery},\ \bibinfo {address} {New York, NY, USA},\ \bibinfo {year} {2024})\ p.\ \bibinfo {pages} {1507–1516}\BibitemShut {NoStop}%
\bibitem [{\citenamefont {Bergamaschi}\ \emph {et~al.}(2024{\natexlab{b}})\citenamefont {Bergamaschi}, \citenamefont {Jeronimo}, \citenamefont {Mittal}, \citenamefont {Srivastava},\ and\ \citenamefont {Tulsiani}}]{bergamaschi2024listdecodablequantumldpc}%
  \BibitemOpen
  \bibfield  {author} {\bibinfo {author} {\bibfnamefont {T.}~\bibnamefont {Bergamaschi}}, \bibinfo {author} {\bibfnamefont {F.~G.}\ \bibnamefont {Jeronimo}}, \bibinfo {author} {\bibfnamefont {T.}~\bibnamefont {Mittal}}, \bibinfo {author} {\bibfnamefont {S.}~\bibnamefont {Srivastava}},\ and\ \bibinfo {author} {\bibfnamefont {M.}~\bibnamefont {Tulsiani}},\ }\bibfield  {title} {\bibinfo {title} {List decodable quantum {LDPC} codes},\ }\href {https://doi.org/10.48550/arXiv.2411.04306} {\bibfield  {journal} {\bibinfo  {journal} {arXiv preprint arXiv:2411.04306}\ } (\bibinfo {year} {2024}{\natexlab{b}})}\BibitemShut {NoStop}%
\bibitem [{\citenamefont {Knill}\ and\ \citenamefont {Laflamme}(1997)}]{PhysRevA.55.900}%
  \BibitemOpen
  \bibfield  {author} {\bibinfo {author} {\bibfnamefont {E.}~\bibnamefont {Knill}}\ and\ \bibinfo {author} {\bibfnamefont {R.}~\bibnamefont {Laflamme}},\ }\bibfield  {title} {\bibinfo {title} {Theory of quantum error-correcting codes},\ }\href {https://doi.org/10.1103/PhysRevA.55.900} {\bibfield  {journal} {\bibinfo  {journal} {Phys. Rev. A}\ }\textbf {\bibinfo {volume} {55}},\ \bibinfo {pages} {900} (\bibinfo {year} {1997})}\BibitemShut {NoStop}%
\bibitem [{\citenamefont {Ji}\ \emph {et~al.}(2018)\citenamefont {Ji}, \citenamefont {Liu},\ and\ \citenamefont {Song}}]{ji2018pseudorandom}%
  \BibitemOpen
  \bibfield  {author} {\bibinfo {author} {\bibfnamefont {Z.}~\bibnamefont {Ji}}, \bibinfo {author} {\bibfnamefont {Y.-K.}\ \bibnamefont {Liu}},\ and\ \bibinfo {author} {\bibfnamefont {F.}~\bibnamefont {Song}},\ }\bibfield  {title} {\bibinfo {title} {Pseudorandom quantum states},\ }in\ \href {https://doi.org/10.1007/978-3-319-96878-0_5} {\emph {\bibinfo {booktitle} {Annual International Cryptology Conference}}}\ (\bibinfo {organization} {Springer},\ \bibinfo {year} {2018})\ pp.\ \bibinfo {pages} {126--152}\BibitemShut {NoStop}%
\bibitem [{\citenamefont {Ma}\ and\ \citenamefont {Huang}(2025)}]{ma2025constructrandomunitaries}%
  \BibitemOpen
  \bibfield  {author} {\bibinfo {author} {\bibfnamefont {F.}~\bibnamefont {Ma}}\ and\ \bibinfo {author} {\bibfnamefont {H.-Y.}\ \bibnamefont {Huang}},\ }\bibfield  {title} {\bibinfo {title} {How to construct random unitaries},\ }in\ \href {https://doi.org/10.1145/3717823.3718254} {\emph {\bibinfo {booktitle} {Proceedings of the 57th Annual ACM Symposium on Theory of Computing}}},\ \bibinfo {series and number} {STOC '25}\ (\bibinfo  {publisher} {Association for Computing Machinery},\ \bibinfo {address} {New York, NY, USA},\ \bibinfo {year} {2025})\ p.\ \bibinfo {pages} {806–809}\BibitemShut {NoStop}%
\bibitem [{\citenamefont {Dennis}\ \emph {et~al.}(2002)\citenamefont {Dennis}, \citenamefont {Kitaev}, \citenamefont {Landahl},\ and\ \citenamefont {Preskill}}]{dennis2002topological}%
  \BibitemOpen
  \bibfield  {author} {\bibinfo {author} {\bibfnamefont {E.}~\bibnamefont {Dennis}}, \bibinfo {author} {\bibfnamefont {A.}~\bibnamefont {Kitaev}}, \bibinfo {author} {\bibfnamefont {A.}~\bibnamefont {Landahl}},\ and\ \bibinfo {author} {\bibfnamefont {J.}~\bibnamefont {Preskill}},\ }\bibfield  {title} {\bibinfo {title} {Topological quantum memory},\ }\href {https://doi.org/10.1063/1.1499754} {\bibfield  {journal} {\bibinfo  {journal} {Journal of Mathematical Physics}\ }\textbf {\bibinfo {volume} {43}},\ \bibinfo {pages} {4452} (\bibinfo {year} {2002})}\BibitemShut {NoStop}%
\bibitem [{\citenamefont {Schuster}\ \emph {et~al.}(2025)\citenamefont {Schuster}, \citenamefont {Haferkamp},\ and\ \citenamefont {Huang}}]{schuster2024randomunitariesextremelylow}%
  \BibitemOpen
  \bibfield  {author} {\bibinfo {author} {\bibfnamefont {T.}~\bibnamefont {Schuster}}, \bibinfo {author} {\bibfnamefont {J.}~\bibnamefont {Haferkamp}},\ and\ \bibinfo {author} {\bibfnamefont {H.-Y.}\ \bibnamefont {Huang}},\ }\bibfield  {title} {\bibinfo {title} {Random unitaries in extremely low depth},\ }\href {https://doi.org/10.1126/science.adv8590} {\bibfield  {journal} {\bibinfo  {journal} {Science}\ }\textbf {\bibinfo {volume} {389}},\ \bibinfo {pages} {92} (\bibinfo {year} {2025})}\BibitemShut {NoStop}%
\bibitem [{\citenamefont {Cleve}\ \emph {et~al.}(2016)\citenamefont {Cleve}, \citenamefont {Leung}, \citenamefont {Liu},\ and\ \citenamefont {Wang}}]{cleve2015near}%
  \BibitemOpen
  \bibfield  {author} {\bibinfo {author} {\bibfnamefont {R.}~\bibnamefont {Cleve}}, \bibinfo {author} {\bibfnamefont {D.}~\bibnamefont {Leung}}, \bibinfo {author} {\bibfnamefont {L.}~\bibnamefont {Liu}},\ and\ \bibinfo {author} {\bibfnamefont {C.}~\bibnamefont {Wang}},\ }\bibfield  {title} {\bibinfo {title} {{Near-linear constructions of exact unitary 2-designs}},\ }\href {https://doi.org/10.26421/QIC16.9-10-1} {\bibfield  {journal} {\bibinfo  {journal} {Quantum Information {\&} Computation}\ }\textbf {\bibinfo {volume} {16}},\ \bibinfo {pages} {721} (\bibinfo {year} {2016})}\BibitemShut {NoStop}%
\bibitem [{\citenamefont {Haug}\ \emph {et~al.}(2025)\citenamefont {Haug}, \citenamefont {Bansal}, \citenamefont {Mok}, \citenamefont {Koh},\ and\ \citenamefont {Bharti}}]{haug2025pseudorandomquantumauthentication}%
  \BibitemOpen
  \bibfield  {author} {\bibinfo {author} {\bibfnamefont {T.}~\bibnamefont {Haug}}, \bibinfo {author} {\bibfnamefont {N.}~\bibnamefont {Bansal}}, \bibinfo {author} {\bibfnamefont {W.-K.}\ \bibnamefont {Mok}}, \bibinfo {author} {\bibfnamefont {D.~E.}\ \bibnamefont {Koh}},\ and\ \bibinfo {author} {\bibfnamefont {K.}~\bibnamefont {Bharti}},\ }\bibfield  {title} {\bibinfo {title} {Pseudorandom quantum authentication},\ }\href {https://doi.org/10.48550/arXiv.2501.00951} {\bibfield  {journal} {\bibinfo  {journal} {arXiv preprint arXiv:2501.00951}\ } (\bibinfo {year} {2025})}\BibitemShut {NoStop}%
\bibitem [{\citenamefont {Barnum}\ \emph {et~al.}(2002)\citenamefont {Barnum}, \citenamefont {Cr{\'e}peau}, \citenamefont {Gottesman}, \citenamefont {Smith},\ and\ \citenamefont {Tapp}}]{barnum2002authentication}%
  \BibitemOpen
  \bibfield  {author} {\bibinfo {author} {\bibfnamefont {H.}~\bibnamefont {Barnum}}, \bibinfo {author} {\bibfnamefont {C.}~\bibnamefont {Cr{\'e}peau}}, \bibinfo {author} {\bibfnamefont {D.}~\bibnamefont {Gottesman}}, \bibinfo {author} {\bibfnamefont {A.}~\bibnamefont {Smith}},\ and\ \bibinfo {author} {\bibfnamefont {A.}~\bibnamefont {Tapp}},\ }\bibfield  {title} {\bibinfo {title} {Authentication of quantum messages},\ }in\ \href {https://doi.org/10.1109/SFCS.2002.1181969} {\emph {\bibinfo {booktitle} {The 43rd Annual IEEE Symposium on Foundations of Computer Science, 2002. Proceedings.}}}\ (\bibinfo {organization} {IEEE},\ \bibinfo {year} {2002})\ pp.\ \bibinfo {pages} {449--458}\BibitemShut {NoStop}%
\bibitem [{\citenamefont {Portmann}(2017)}]{portmann2017}%
  \BibitemOpen
  \bibfield  {author} {\bibinfo {author} {\bibfnamefont {C.}~\bibnamefont {Portmann}},\ }\href {https://eprint.iacr.org/2017/119} {\bibinfo {title} {Quantum authentication with key recycling}},\ \bibinfo {howpublished} {Cryptology {ePrint} Archive, Paper 2017/119} (\bibinfo {year} {2017})\BibitemShut {NoStop}%
\bibitem [{\citenamefont {Dulek}\ and\ \citenamefont {Speelman}(2018)}]{dulek2018quantum}%
  \BibitemOpen
  \bibfield  {author} {\bibinfo {author} {\bibfnamefont {Y.}~\bibnamefont {Dulek}}\ and\ \bibinfo {author} {\bibfnamefont {F.}~\bibnamefont {Speelman}},\ }\bibfield  {title} {\bibinfo {title} {{Quantum Ciphertext Authentication and Key Recycling with the Trap Code}},\ }in\ \href {https://doi.org/10.4230/LIPIcs.TQC.2018.1} {\emph {\bibinfo {booktitle} {13th Conference on the Theory of Quantum Computation, Communication and Cryptography (TQC 2018)}}},\ \bibinfo {series} {Leibniz International Proceedings in Informatics (LIPIcs)}, Vol.\ \bibinfo {volume} {111},\ \bibinfo {editor} {edited by\ \bibinfo {editor} {\bibfnamefont {S.}~\bibnamefont {Jeffery}}}\ (\bibinfo  {publisher} {Schloss Dagstuhl -- Leibniz-Zentrum f{\"u}r Informatik},\ \bibinfo {address} {Dagstuhl, Germany},\ \bibinfo {year} {2018})\ pp.\ \bibinfo {pages} {1:1--1:17}\BibitemShut {NoStop}%
\bibitem [{\citenamefont {Aly}\ \emph {et~al.}(2006)\citenamefont {Aly}, \citenamefont {Klappenecker},\ and\ \citenamefont {Sarvepalli}}]{aly2006subsystem}%
  \BibitemOpen
  \bibfield  {author} {\bibinfo {author} {\bibfnamefont {S.~A.}\ \bibnamefont {Aly}}, \bibinfo {author} {\bibfnamefont {A.}~\bibnamefont {Klappenecker}},\ and\ \bibinfo {author} {\bibfnamefont {P.~K.}\ \bibnamefont {Sarvepalli}},\ }\bibfield  {title} {\bibinfo {title} {Subsystem codes},\ }\href {https://doi.org/10.48550/arXiv.quant-ph/0610153} {\bibfield  {journal} {\bibinfo  {journal} {arXiv preprint quant-ph/0610153}\ } (\bibinfo {year} {2006})}\BibitemShut {NoStop}%
\bibitem [{\citenamefont {Yoder}(2019)}]{yoder2019optimal}%
  \BibitemOpen
  \bibfield  {author} {\bibinfo {author} {\bibfnamefont {T.~J.}\ \bibnamefont {Yoder}},\ }\bibfield  {title} {\bibinfo {title} {Optimal quantum subsystem codes in two dimensions},\ }\href {https://doi.org/10.1103/PhysRevA.99.052333} {\bibfield  {journal} {\bibinfo  {journal} {Phys. Rev. A}\ }\textbf {\bibinfo {volume} {99}},\ \bibinfo {pages} {052333} (\bibinfo {year} {2019})}\BibitemShut {NoStop}%
\bibitem [{\citenamefont {Klappenecker}\ and\ \citenamefont {Sarvepalli}(2007)}]{klappenecker2007subsystem}%
  \BibitemOpen
  \bibfield  {author} {\bibinfo {author} {\bibfnamefont {A.}~\bibnamefont {Klappenecker}}\ and\ \bibinfo {author} {\bibfnamefont {P.~K.}\ \bibnamefont {Sarvepalli}},\ }\bibfield  {title} {\bibinfo {title} {On subsystem codes beating the quantum {H}amming or {S}ingleton bound},\ }\href {https://doi.org/10.1098/rspa.2007.0028} {\bibfield  {journal} {\bibinfo  {journal} {Proceedings of the Royal Society A: Mathematical, Physical and Engineering Sciences}\ }\textbf {\bibinfo {volume} {463}},\ \bibinfo {pages} {2887} (\bibinfo {year} {2007})}\BibitemShut {NoStop}%
\bibitem [{\citenamefont {Bacon}\ and\ \citenamefont {Casaccino}(2006)}]{bacon2006quantum}%
  \BibitemOpen
  \bibfield  {author} {\bibinfo {author} {\bibfnamefont {D.}~\bibnamefont {Bacon}}\ and\ \bibinfo {author} {\bibfnamefont {A.}~\bibnamefont {Casaccino}},\ }\bibfield  {title} {\bibinfo {title} {Quantum error correcting subsystem codes from two classical linear codes},\ }\href {https://doi.org/10.48550/arXiv.quant-ph/0610088} {\bibfield  {journal} {\bibinfo  {journal} {arXiv preprint quant-ph/0610088}\ } (\bibinfo {year} {2006})}\BibitemShut {NoStop}%
\bibitem [{\citenamefont {Davydova}\ \emph {et~al.}(2023)\citenamefont {Davydova}, \citenamefont {Tantivasadakarn},\ and\ \citenamefont {Balasubramanian}}]{davydova2023floquet}%
  \BibitemOpen
  \bibfield  {author} {\bibinfo {author} {\bibfnamefont {M.}~\bibnamefont {Davydova}}, \bibinfo {author} {\bibfnamefont {N.}~\bibnamefont {Tantivasadakarn}},\ and\ \bibinfo {author} {\bibfnamefont {S.}~\bibnamefont {Balasubramanian}},\ }\bibfield  {title} {\bibinfo {title} {Floquet codes without parent subsystem codes},\ }\href {https://doi.org/10.1103/PRXQuantum.4.020341} {\bibfield  {journal} {\bibinfo  {journal} {PRX Quantum}\ }\textbf {\bibinfo {volume} {4}},\ \bibinfo {pages} {020341} (\bibinfo {year} {2023})}\BibitemShut {NoStop}%
\bibitem [{\citenamefont {Tanggara}\ \emph {et~al.}(2024{\natexlab{a}})\citenamefont {Tanggara}, \citenamefont {Gu},\ and\ \citenamefont {Bharti}}]{tanggara2024simple}%
  \BibitemOpen
  \bibfield  {author} {\bibinfo {author} {\bibfnamefont {A.}~\bibnamefont {Tanggara}}, \bibinfo {author} {\bibfnamefont {M.}~\bibnamefont {Gu}},\ and\ \bibinfo {author} {\bibfnamefont {K.}~\bibnamefont {Bharti}},\ }\bibfield  {title} {\bibinfo {title} {Simple construction of qudit {F}loquet codes on a family of lattices},\ }\href {https://doi.org/10.48550/arXiv.2410.02022} {\bibfield  {journal} {\bibinfo  {journal} {arXiv preprint arXiv:2410.02022}\ } (\bibinfo {year} {2024}{\natexlab{a}})}\BibitemShut {NoStop}%
\bibitem [{\citenamefont {Fahimniya}\ \emph {et~al.}(2025)\citenamefont {Fahimniya}, \citenamefont {Dehghani}, \citenamefont {Bharti}, \citenamefont {Mathew}, \citenamefont {Koll{\'{a}}r}, \citenamefont {Gorshkov},\ and\ \citenamefont {Gullans}}]{fahimniya2023fault}%
  \BibitemOpen
  \bibfield  {author} {\bibinfo {author} {\bibfnamefont {A.}~\bibnamefont {Fahimniya}}, \bibinfo {author} {\bibfnamefont {H.}~\bibnamefont {Dehghani}}, \bibinfo {author} {\bibfnamefont {K.}~\bibnamefont {Bharti}}, \bibinfo {author} {\bibfnamefont {S.}~\bibnamefont {Mathew}}, \bibinfo {author} {\bibfnamefont {A.~J.}\ \bibnamefont {Koll{\'{a}}r}}, \bibinfo {author} {\bibfnamefont {A.~V.}\ \bibnamefont {Gorshkov}},\ and\ \bibinfo {author} {\bibfnamefont {M.~J.}\ \bibnamefont {Gullans}},\ }\bibfield  {title} {\bibinfo {title} {Fault-tolerant hyperbolic {F}loquet quantum error correcting codes},\ }\href {https://doi.org/10.22331/q-2025-09-05-1849} {\bibfield  {journal} {\bibinfo  {journal} {{Quantum}}\ }\textbf {\bibinfo {volume} {9}},\ \bibinfo {pages} {1849} (\bibinfo {year} {2025})}\BibitemShut {NoStop}%
\bibitem [{\citenamefont {Zhang}\ \emph {et~al.}(2023)\citenamefont {Zhang}, \citenamefont {Aasen},\ and\ \citenamefont {Vijay}}]{zhang2023x}%
  \BibitemOpen
  \bibfield  {author} {\bibinfo {author} {\bibfnamefont {Z.}~\bibnamefont {Zhang}}, \bibinfo {author} {\bibfnamefont {D.}~\bibnamefont {Aasen}},\ and\ \bibinfo {author} {\bibfnamefont {S.}~\bibnamefont {Vijay}},\ }\bibfield  {title} {\bibinfo {title} {$x$-cube {F}loquet code: A dynamical quantum error correcting code with a subextensive number of logical qubits},\ }\href {https://doi.org/10.1103/PhysRevB.108.205116} {\bibfield  {journal} {\bibinfo  {journal} {Phys. Rev. B}\ }\textbf {\bibinfo {volume} {108}},\ \bibinfo {pages} {205116} (\bibinfo {year} {2023})}\BibitemShut {NoStop}%
\bibitem [{\citenamefont {Bacon}\ \emph {et~al.}(2017)\citenamefont {Bacon}, \citenamefont {Flammia}, \citenamefont {Harrow},\ and\ \citenamefont {Shi}}]{bacon2017sparse}%
  \BibitemOpen
  \bibfield  {author} {\bibinfo {author} {\bibfnamefont {D.}~\bibnamefont {Bacon}}, \bibinfo {author} {\bibfnamefont {S.~T.}\ \bibnamefont {Flammia}}, \bibinfo {author} {\bibfnamefont {A.~W.}\ \bibnamefont {Harrow}},\ and\ \bibinfo {author} {\bibfnamefont {J.}~\bibnamefont {Shi}},\ }\bibfield  {title} {\bibinfo {title} {Sparse quantum codes from quantum circuits},\ }\href {https://doi.org/10.1109/TIT.2017.2663199} {\bibfield  {journal} {\bibinfo  {journal} {IEEE Transactions on Information Theory}\ }\textbf {\bibinfo {volume} {63}},\ \bibinfo {pages} {2464} (\bibinfo {year} {2017})}\BibitemShut {NoStop}%
\bibitem [{\citenamefont {Gottesman}(2022)}]{gottesman2022opportunities}%
  \BibitemOpen
  \bibfield  {author} {\bibinfo {author} {\bibfnamefont {D.}~\bibnamefont {Gottesman}},\ }\bibfield  {title} {\bibinfo {title} {Opportunities and challenges in fault-tolerant quantum computation},\ }\href {https://doi.org/10.48550/arXiv.2210.15844} {\bibfield  {journal} {\bibinfo  {journal} {arXiv preprint arXiv:2210.15844}\ } (\bibinfo {year} {2022})}\BibitemShut {NoStop}%
\bibitem [{\citenamefont {Delfosse}\ and\ \citenamefont {Paetznick}(2023)}]{delfosse2023spacetime}%
  \BibitemOpen
  \bibfield  {author} {\bibinfo {author} {\bibfnamefont {N.}~\bibnamefont {Delfosse}}\ and\ \bibinfo {author} {\bibfnamefont {A.}~\bibnamefont {Paetznick}},\ }\bibfield  {title} {\bibinfo {title} {Spacetime codes of {C}lifford circuits},\ }\href {https://doi.org/10.48550/arXiv.2304.05943} {\bibfield  {journal} {\bibinfo  {journal} {arXiv preprint arXiv:2304.05943}\ } (\bibinfo {year} {2023})}\BibitemShut {NoStop}%
\bibitem [{\citenamefont {Tanggara}\ \emph {et~al.}(2024{\natexlab{b}})\citenamefont {Tanggara}, \citenamefont {Gu},\ and\ \citenamefont {Bharti}}]{tanggara2024strategic}%
  \BibitemOpen
  \bibfield  {author} {\bibinfo {author} {\bibfnamefont {A.}~\bibnamefont {Tanggara}}, \bibinfo {author} {\bibfnamefont {M.}~\bibnamefont {Gu}},\ and\ \bibinfo {author} {\bibfnamefont {K.}~\bibnamefont {Bharti}},\ }\bibfield  {title} {\bibinfo {title} {Strategic code: A unified spatio-temporal framework for quantum error-correction},\ }\href {https://doi.org/10.48550/arXiv.2405.17567} {\bibfield  {journal} {\bibinfo  {journal} {arXiv preprint arXiv:2405.17567}\ } (\bibinfo {year} {2024}{\natexlab{b}})}\BibitemShut {NoStop}%
\bibitem [{\citenamefont {Hinsche}\ \emph {et~al.}(2025)\citenamefont {Hinsche}, \citenamefont {Eisert},\ and\ \citenamefont {Carrasco}}]{hinsche2025abelianstatehiddensubgroup}%
  \BibitemOpen
  \bibfield  {author} {\bibinfo {author} {\bibfnamefont {M.}~\bibnamefont {Hinsche}}, \bibinfo {author} {\bibfnamefont {J.}~\bibnamefont {Eisert}},\ and\ \bibinfo {author} {\bibfnamefont {J.}~\bibnamefont {Carrasco}},\ }\bibfield  {title} {\bibinfo {title} {The abelian state hidden subgroup problem: Learning stabilizer groups and beyond},\ }\href {https://doi.org/10.48550/arXiv.2505.15770} {\bibfield  {journal} {\bibinfo  {journal} {arXiv preprint arXiv:2505.15770}\ } (\bibinfo {year} {2025})}\BibitemShut {NoStop}%
\bibitem [{\citenamefont {Cr{\'e}peau}\ \emph {et~al.}(2005)\citenamefont {Cr{\'e}peau}, \citenamefont {Gottesman},\ and\ \citenamefont {Smith}}]{crepeau2005approximatequantumerrorcorrectingcodes}%
  \BibitemOpen
  \bibfield  {author} {\bibinfo {author} {\bibfnamefont {C.}~\bibnamefont {Cr{\'e}peau}}, \bibinfo {author} {\bibfnamefont {D.}~\bibnamefont {Gottesman}},\ and\ \bibinfo {author} {\bibfnamefont {A.}~\bibnamefont {Smith}},\ }\bibfield  {title} {\bibinfo {title} {Approximate quantum error-correcting codes and secret sharing schemes},\ }in\ \href {https://doi.org/10.1007/11426639_17} {\emph {\bibinfo {booktitle} {Advances in Cryptology -- EUROCRYPT 2005}}},\ \bibinfo {editor} {edited by\ \bibinfo {editor} {\bibfnamefont {R.}~\bibnamefont {Cramer}}}\ (\bibinfo  {publisher} {Springer Berlin Heidelberg},\ \bibinfo {address} {Berlin, Heidelberg},\ \bibinfo {year} {2005})\ pp.\ \bibinfo {pages} {285--301}\BibitemShut {NoStop}%
\bibitem [{\citenamefont {Nielsen}\ and\ \citenamefont {Chuang}(2000)}]{nielsen00}%
  \BibitemOpen
  \bibfield  {author} {\bibinfo {author} {\bibfnamefont {M.~A.}\ \bibnamefont {Nielsen}}\ and\ \bibinfo {author} {\bibfnamefont {I.~L.}\ \bibnamefont {Chuang}},\ }\href {https://doi.org/10.1017/CBO9780511976667} {\emph {\bibinfo {title} {Quantum Computation and Quantum Information}}}\ (\bibinfo  {publisher} {Cambridge University Press},\ \bibinfo {year} {2000})\BibitemShut {NoStop}%
\bibitem [{\citenamefont {Mele}(2024)}]{mele2024introduction}%
  \BibitemOpen
  \bibfield  {author} {\bibinfo {author} {\bibfnamefont {A.~A.}\ \bibnamefont {Mele}},\ }\bibfield  {title} {\bibinfo {title} {Introduction to {H}aar {M}easure {T}ools in {Q}uantum {I}nformation: {A} {B}eginner's {T}utorial},\ }\href {https://doi.org/10.22331/q-2024-05-08-1340} {\bibfield  {journal} {\bibinfo  {journal} {{Quantum}}\ }\textbf {\bibinfo {volume} {8}},\ \bibinfo {pages} {1340} (\bibinfo {year} {2024})}\BibitemShut {NoStop}%
\bibitem [{\citenamefont {Fukuda}\ \emph {et~al.}(2019)\citenamefont {Fukuda}, \citenamefont {König},\ and\ \citenamefont {Nechita}}]{Fukuda_2019}%
  \BibitemOpen
  \bibfield  {author} {\bibinfo {author} {\bibfnamefont {M.}~\bibnamefont {Fukuda}}, \bibinfo {author} {\bibfnamefont {R.}~\bibnamefont {König}},\ and\ \bibinfo {author} {\bibfnamefont {I.}~\bibnamefont {Nechita}},\ }\bibfield  {title} {\bibinfo {title} {{RTNI}—a symbolic integrator for {H}aar-random tensor networks},\ }\href {https://doi.org/10.1088/1751-8121/ab434b} {\bibfield  {journal} {\bibinfo  {journal} {Journal of Physics A: Mathematical and Theoretical}\ }\textbf {\bibinfo {volume} {52}},\ \bibinfo {pages} {425303} (\bibinfo {year} {2019})}\BibitemShut {NoStop}%
\end{thebibliography}%

\onecolumngrid
\section*{End Matter}
\label{endmatter}
\twocolumngrid

\prlsection{Definition of pseudorandom unitaries}
We denote the set of pure states over a Hilbert space $\mathcal{H}$ using $S(\mathcal{H})$. Further, we consider the set of unitary operators on the Hilbert space and its associated uniform measure, i.e. the Haar measure. 
However, unitaries drawn according to the Haar measure, denoted as Haar random unitaries, cannot be implemented efficiently. 
Recently, pseudorandom unitaries (PRUs) have been proposed~\cite{ji2018pseudorandom, ma2025constructrandomunitaries} that alleviate this problem by introducing  assumptions on the computational complexities of the adversary:  
PRUs are indistinguishable from Haar random unitaries for any algorithm that can be efficiently implemented (i.e. polynomial circuit depth and qubits), yet can be efficiently implemented. 
To put it formally, we firstly define the negligible functions:
\begin{definition}[Negligible function]
    A {positive real-valued function $\mu:\mathbb{N} \to \mathbb R$} is negligible if and only if $\forall c \in \mathbb{N}$, $\exists n_0 \in \mathbb{N}$ such that $\forall n > n_0$, $\mu(n) < n^{-c}$.
\end{definition}
Then, the pseudorandom unitaries can be defined as follows.
\begin{definition}[Pseudorandom unitaries (PRUs)~\cite{ji2018pseudorandom}]\label{def:PRU}
Let $\kappa$ be the security parameter, and let $\mathcal{H}$ and $\mathcal{K}$ be the Hilbert space and the key space, respectively, both dependent on $\kappa$. We call a family of unitaries pseudorandom if:
\begin{enumerate}
\item Efficiently preparable: There is a quantum polynomial time algorithm $M$ such that for all $\lambda \in \mathcal{K}$ and $|\psi\rangle \in S(\mathcal{H})$, $M(\lambda, |\psi\rangle) = U_{\lambda} |\psi\rangle$.

\item Computational Indistinguishability: For any quantum polynomial time (QPT) algorithm $\mathcal{A}$ making ${O}(\poly(\kappa))$ queries to the PRUs and the Haar random unitaries,
\begin{equation}
\label{pru}
    \left\vert\Pr_{\lambda \leftarrow \mathcal{K}}[\mathcal{A}^{U_{\lambda}}(1^{\kappa}) = 1] - \Pr_{U \leftarrow \mu_H}[\mathcal{A}^{U}(1^{\kappa}) = 1] \right\vert \leq \negl(\lambda),
\end{equation}
i.e. any QPT algorithm cannot distinguish between PRUs and Haar random unitaries. 
\end{enumerate}
\end{definition}

\prlsection{Alternate Definition of Quantum list decoding} 

\begin{definition}[Alternate definition of Quantum list decoding codes]\label{def:QLD_2}
    Let $Q_L$ be an $\dsl n,k,d \dsr$ quantum stabilizer code and let $\mathcal{E} \subset P_{n}$ be a set of errors. Then $Q_L$ is said to be $L-$QLD (quantum list decodable) if there exists a decoding map $\mathcal{D}_{Q_L}$ such that for all $E \in \mathcal{E}$ and for all encoded states $|\bar{\psi}\rangle = V_{Q_L}|\psi\rangle \in Q_L$, we have 
    \begin{equation}
        \mathcal{D}_{Q_L}(E |\bar{\psi}\rangle \langle\bar{\psi}| E^{\dagger}) = \sum_{i=1}^{l(s)} \frac{1}{l(s)} P_i^s |{\psi}\rangle \langle{\psi}| P_i^{s, \dagger} 
    \end{equation}
    where $P_i^s \in P_k$ where $s$ is the syndrome of the error $E$, and there are $l(s)$ such operators for the syndrome $s$ such that $\max_s l(s) \leq L$.
\end{definition}
This definition is motivated by the QLD definition introduced in~\cite{leungsmith}. The definition considered in the maintext is also similar; however, there we consider the list errors on the code-space of QLD, while the above definition considers list errors on the logical space.  We will show that for particular choices of $P_i^s$, this decoding definition can be obtained from the one considered in Def~\ref{def:QLD}.

\begin{lemma}\label{lemma:equiv_defn}
    The QLD code $Q_L'$ defined according to Def.~\ref{def:QLD_2} can be obtained from QLD code $Q_L$ with the same code-space defined according to Def.~\ref{def:QLD} if the list errors of both codes are related as
    \begin{equation}
        C_Q^{\dagger} E^{s}_i C_Q  = P_i^s \otimes A_i^s
    \end{equation}
    where $C_Q$ is the Clifford unitary corresponding to the common code space, $A^s_i$ is an $m$-qubit Pauli operator, and $E^{s}_i$ and $P_i^s$ are the list errors of $Q_L$ and $Q_L'$, respectively.
\end{lemma}
\begin{proof}
    It can be shown by constructing the decoder for $Q_L'$ from that of $Q_L$ and vice versa. Note that the decoder $\mathcal{D}_{Q_L}$ gives,
    \begin{equation}
        \mathcal{D}_{Q_L}(E |\bar\psi\rangle \langle \bar\psi| E^{\dagger}) = \sum_i^{l(s)} \frac{1}{l(s)} E^{s}_i |\bar\psi\rangle \langle \bar\psi|E_i^{s, \dagger}
    \end{equation}
    for all $E \in \mathcal{E}$ and $s$ being the syndrome of the error $E$. Then, define $\mathcal{D}_{Q_L'} = \mathcal{C} \circ \mathcal{D}_{Q_L}$ where $\mathcal{C}(\rho) = \tr_m( C_Q^{\dagger} \rho C_Q)$ where $\tr_m(.)$ is the partial trace on $m$ ancilla qubits. Then, it is easy to see that, 
    \begin{align*}
           &\mathcal{D}_{Q_L'}(E |\bar\psi\rangle \langle \bar\psi| E^{\dagger}) = 
           \tr_m\left(\sum_i^{l(s)}\frac{1}{l(s)} C_Q^{\dagger} E^{s}_i |\bar\psi\rangle \langle \bar\psi|E_i^{s, \dagger} C_Q \right).
    \end{align*}
    Now we can write $|\bar \psi\rangle = C_Q (|\psi\rangle \otimes |0\rangle)$. Then, putting $C_Q^{\dagger} E^{s}_i C_Q  = P_i^s \otimes A_i^s$, we get
    \begin{equation}
        \begin{gathered}
        \mathcal{D}_{Q_L'}(E |\bar\psi\rangle \langle \bar\psi| E^{\dagger}) =  \\ \tr_m \left(\frac{1}{l(s)}\sum_i^{l(s)} P_i^s |{\psi}\rangle \langle{\psi}| P_i^{s, \dagger} \otimes A_i^s |0\rangle \langle 0| A_i^{s, \dagger}\right).
        \end{gathered}
    \end{equation}
    Thus, taking the partial trace on ancilla gives the required decoding as in Def.~\ref{def:QLD_2}. 
 \end{proof}
Thus, from the perspective of decoding, we can first apply the inverse of encoding isometry to obtain a list of errors acting on the logical space. We will use Def.~\ref{def:QLD} to show the Knill-Laflamme conditions for the quantum list decoding codes, while we will use Def.~\ref{def:QLD_2} in our protocol to enable unambiguous decoding using a pseudorandom unitary with security against adversarial noise.

\prlsection{Purity testing codes}
In keeping with previous literature~\cite{barnum2002authentication, portmann2017, dulek2018quantum}, we stick to $\dsl n, n-m \dsr$ codes $Q$ with associated unitary $V_Q$ which maps $n-m$ qubits state $|\psi\rangle$ to $V_Q(|\psi\rangle \langle \psi| \otimes |0\rangle \langle 0|^{\otimes m})V_Q^{\dagger}$ %
for which a Pauli error on code space is equivalent to Pauli errors on the logical and ancillary qubits. In other words, for every Pauli $P$ of the code space, there exists another $n$ qubit Pauli matrix $P^{\prime}$, such that $V_Q^{\dagger} P V_Q = P^{\prime}$ upto a phase. A typical example of such codes are stabilizer codes, for which the encoding unitary $C_Q$ is a Clifford and which maps a Pauli to some other Pauli. 

\begin{definition}[Purity testing codes~\cite{barnum2002authentication, portmann2017, dulek2018quantum}]\label{def:ptc}
    A $\lambda$-keyed family of codes $\{Q_{\lambda}\}_{\lambda}$ with associated unitaries $\{V_{Q_\lambda}\}_{\lambda}$ which maps an $n-m$ qubit logical state $|\psi\rangle$ to $n$ qubit state $V_{Q_\lambda}(|\psi\rangle \otimes |0\rangle)$ is called an $\epsilon$-purity testing code if for any Pauli $P \in \mathcal{P}_n$, 

    \begin{equation}
        \Pr_{\lambda}[V_{Q_\lambda}^{\dagger} P V_{Q_\lambda} \in (\mathcal{P}_{n-m} \setminus I^{\otimes (n-m)}) \otimes \{I, Z\}^{\otimes m}] \leq \epsilon\,.
    \end{equation}
\end{definition}

In other words, if a Pauli error acts non-trivially on the logical state, then it can detected by measuring the ancillary qubits with high probability. $2$-designs form purity testing codes~\cite{barnum2002authentication}. Intuitively, $2$-designs create high entanglement between logical qubits and ancilla; thus, any non-trivial error on the logical state also affects the ancilla in a non-trivial fashion. The purity testing codes have also been used to construct quantum authentication schemes~\cite{barnum2002authentication}. %

\prlsection{Adversarial Noise}
While in standard decoding, the noise is assumed to be uncorrelated and randomly chosen, this is often a highly simplified assumption. We consider more general noise models where the noise is chosen adversarially and adaptively over multiple rounds, where we now define such an adversarial noise model:
\begin{definition}[Adversarial Noise]\label{def:advnoise}
The $n$-qubit $(p_\text{e},t)$ adversarial noise $\mathcal{E}$ applies channel $\mathcal{A}$ with Kraus operators of the form $A_i = \sum_{E \in \mathcal{P}_n^{p_\text{e}n}}\alpha^i_E E$ where $\mathcal{P}_n^{p_\text{e}n}$ are Pauli errors with weight at most $p_\text{e}n$. 
The particular choice of $\alpha^i_E$ can be changed adaptively by the adversary in each round.
Further, we assume that before the communication rounds, the adversary has had access to $t$ copies of the code state. The adversary can use any efficiently implementable quantum algorithm $\mathcal{B}$ to learn from the $t$ states.
\end{definition}
A physical example of such a noise model would be an adversary that managed to gain access to $t$ copies of the code state. For example, the adversary could gain the copies via loss that occurred during the communication.
Then, the adversary learns from the $t$ copies in order to find a specific error of weight $p_\text{e}n$ that is not correctable. The adversary can then use this error to interrupt the communication with the code.%

We highlight that this model is more general than previously considered noise models~\cite{leungsmith, bergamaschi2022approachingquantumsingletonbound}. %
We also assume that the adversary has bounded computational power, i.e. it can only run quantum algorithms with polynomial circuit depth and qubits. In contrast, previous works assumed adversaries with unbounded computational power~\cite{leungsmith, bergamaschi2022approachingquantumsingletonbound}.

It is easy to see that allowing a non-zero $t$ allows the adversary to break the security provided by local indistinguishability. As an example, we outline an attack that allows the adversary to learn the logical state even when the stabilizer code is unknown for $t=O(n)$. The adversary can learn the underlying stabilizer group~\cite{hinsche2025abelianstatehiddensubgroup} and learn the encoding isometry and key $\lambda$. With this information, the adversary can find a weight $w$ channel that is not correctable by the approximate quantum error correction code. %

We highlight that previous constructions of private approximate quantum error correction codes from quantum list decoding codes~\cite{bergamaschi2022approachingquantumsingletonbound} are not key-reusable since, as purity testing codes are vulnerable to our adversarial noise model if the key $\lambda$ is re-used multiple times.
In contrast, our codes are secure even when the same key is used $t=\text{poly}(n)$ times, as learning the key is inefficient. This follows directly from the definition of  PRUs that are used within our code. %

\prlsection{Cryptographic Approximate Quantum error correcting codes}
\begin{definition}[Cryptographic approximate quantum error correcting code]\label{def:private_AQECC}
    A cryptographic approximate quantum error correcting code  is a pair of algorithms $(\textit{Enc}, \textit{Dec})$,
    \begin{enumerate}
        \item \textit{Enc} takes input as key $\lambda$ and $n$-qubit logical state $|\psi\rangle \langle \psi|$ and outputs $n^{\prime}$ encoded qubits.
        \item  \textit{Dec} takes input as the same key $\lambda$ and $n^{\prime}$ noisy state $\rho$ and outputs a $n$-qubit state $\sigma$.
    \end{enumerate}
    We call a cryptographic approximate quantum error correcting code to be $(\mathcal{A}, \delta)$-secure if
    \begin{equation}
        \sup_{|\psi\rangle} \Eset{}[ \mathrm{TD}(\textit{Dec}_\lambda \circ \mathcal{A} \circ \textit{Enc}_\lambda (|\psi\rangle), |\psi\rangle)] \leq \delta
    \end{equation}
    where the expectation is over both keys and randomness in the algorithm, and where $\mathcal{A}$ is some efficient adversarial channel, and we denote the input $\lambda$ to the algorithms using a subscript.
\end{definition}
We call the scheme key-reusable if the key is reused for $\poly(n)$ rounds of communication. If we take the adversary $\mathcal{A}$ to be as defined in Def.~\ref{def:advnoise}, then it includes security under computational assumptions, i.e. the adversary is assumed to have bounded computational power along with full access to $t$ states, while previously proposed private approximate cryptographic error correction codes~\cite{bergamaschi2022approachingquantumsingletonbound,crepeau2005approximatequantumerrorcorrectingcodes} assume a statistical adversary with $t = 0$ (no full state access). These computational assumptions allow for codes with stronger security, e.g. the key can be re-used over multiple communication rounds even against an adversary of Def.~\ref{def:advnoise}.

\clearpage
\newpage

\let\addcontentsline\oldaddcontentsline

\appendix

\onecolumngrid
\newpage 

\setcounter{secnumdepth}{2}
\setcounter{equation}{0}
\setcounter{figure}{0}
\setcounter{section}{0}

\renewcommand{\thesection}{\Alph{section}}
\renewcommand{\thesubsection}{\arabic{subsection}}
\renewcommand*{\theHsection}{\thesection}

\clearpage
\begin{center}

\textbf{\large \SMLong{}}
\end{center}
\setcounter{equation}{0}
\setcounter{figure}{0}
\setcounter{table}{0}

\makeatletter

\renewcommand{\thefigure}{S\arabic{figure}}

We provide proofs and additional details supporting the claims in the main text.

\makeatletter
\@starttoc{toc}

\makeatother

\section{Proof of Knill-Laflamme conditions for quantum list decoding}\label{appendix_KL_conditions_proof}

To show the Knill-Laflamme conditions for quantum list decoding, it is important to be exact about the definition that we start with before we can derive the necessary and sufficient conditions. We will use the Def.~\ref{def:QLD} to derive the Knill-Laflamme conditions, i.e. we start with a stabilizer code with the correctable set of errors $\tilde{\mathcal{E}}$ which are Pauli errors of weight up to $\lfloor{(d-1)}/2\rfloor$. Due to this underlying stabilizer code structure, the usual Knill-Laflamme conditions for the error set $\tilde{\mathcal{E}}$ hold. Moreover, we can define the syndrome of every Pauli operator. We then derive the Knill-Laflamme conditions for quantum list decoding, which apply to a larger set of errors $\mathcal{E} \supset \tilde{\mathcal{E}}$.

\begin{definition}
\label{qld_def}[Quantum list decoding] We say that a stabilizer code $Q$ is $L$-QLD for an error set $\mathcal{E}$ if for any error $E \in \mathcal{E}$ and code space state $\Pi_Q\rho\Pi_Q$, we have a decoding operation $\mathcal{R}$ such that 
\begin{equation}\label{qld_def_eqn}
    \mathcal{R}(E(\Pi_Q\rho \Pi_Q)E^{\dag}) = \sum_{E_{i} \in E\mathcal{N}(Q) \cap \mathcal{E}} c(s) E_{i}(\Pi_Q\rho \Pi_Q)E_{i}^{\dag}\,.
\end{equation}
Here,
\begin{equation}
    E\mathcal{N}(Q) = \{EN| N \in \mathcal{N}(Q)\}
\end{equation}
where $\mathcal{N}(Q)$ is the normaliser of the stabilizer code, $\Pi_Q$ the projector onto the code space, $c(s)>0$ is a normalisation factor, and $E_i$ are (at most) $L$ list errors for the error $E$, that is, the list errors with the same syndrome as $E$. 

\end{definition}
With this in mind, we can posit the following Knill-Laflamme condition for the overall list code:

\begin{theorem}[Knill-Laflamme Conditions for QLD]\label{thm:KL_for_QLD}
Let $Q$ be a $L$-$\mathrm{QLD}$ code for error set $\mathcal{E}= \{E_{i}\}_i$, and let \texttt{Syn} be the associated syndrome operator. Let $\Pi_Q$ be the projector onto the code space. Then, when two errors $E$ and $F$ in $\mathcal{E}$ have different syndromes, we have
\begin{equation}
    \Pi_QE^{\dag}F\Pi_Q = 0\,.
\end{equation}
In contrast, when $E$ and $F$ have the same syndrome $s$, then we have
\begin{equation}
    \Pi_Q E^{\dag} F\Pi_Q = c(s) \alpha(E,F)\Pi_Q,
\end{equation}
where $\alpha(E,F) = \sum_{k,i,j} U^{*}(E)_{ki}\Pi_Q E_{i}^{\dag}U(F)_{kj}E_{j}$ depend on $E,F$ with $U(E)$ and $U(F)$ being unitary matrices and $c(s) = 1/l(s)$ is the inverse of number of list errors corresponding to syndrome $s$.
\end{theorem}
\begin{proof}
    (Necessary). First, we consider the case with two different syndromes $s$ and $s'$. Let $X = \Pi_QE_{i}E^{\dag}_{j}\Pi_Q$. There exists a stabilizer generator $g$ such that $gE_{i} = (-1)^{s(g)}E_{i}g$,  $gE_{j} = (-1)^{s'(g)}E_{j}g$, $s(g) \neq s'(g)$. Now $gX = X$, but $gX = g\Pi_Q E_{i}E^{\dag}_{j}\Pi_Q = X(-1)^{s'(g) +s(g)}$, so $X = 0$. 
    
    Next, we consider the case where we have two errors with the same syndrome $s$. The approach that we take is very similar to the usual Knill-Laflamme conditions. For error $F$, we must have that $\mathcal{R}(F \Pi_Q \rho \Pi_Q F^{\dag}) = \sum_{i} c(s)E_{i}(\Pi_Q \rho \Pi_Q)E_{i}^{\dag}$, where the sum over $i$ is exactly over the set of list errors as in Def.~\ref{def:QLD_2}. This means that recovery channel $\mathcal{R}$ corrects an error $F$ up to a list of errors $\{E_i\}_{i=1}^\ell$, with $\ell\leq L$, that depend on syndrome $s$.  Let $\{R_{k}\}_k$ be the set of Kraus operators that implement recovery channel $\mathcal{R}(\rho)=\sum_k R_k \rho R_k^\dagger$  with $\sum_k R_k^\dagger R_k=I$. 
    For error $F$, we have then
    \begin{equation}
        \sum_{k} R_{k}F \Pi_Q \rho \Pi_Q F^{\dag}R_{k}^{\dag} = \sum_{i} c(s)E_{i}\Pi_Q \rho \Pi_Q E_{i}^{\dag}
    \end{equation}
     and so $R_{k}F\Pi_Q = \sum_i\sqrt{c(s)} U(F)_{ki}E_{i}\Pi_Q$~\cite{nielsen00} for some unitary $U(F)$ relating the two Kraus operator decompositions. Similarly, for error $E$ (with same syndrome as $F$), we have the same procedure, except that the unitary $U(E)$ can be different, i.e. $R_{k}E\Pi_Q = \sum_i \sqrt{c(s)} U(E)_{ki}E_{i}\Pi_Q$. Taking the adjoint of this equation, we get $\Pi_QE^{\dag}R_{k}^{\dag} = \sum_i\sqrt{c(s)} U(E)^{*}_{ki}\Pi_QE_{i}^{\dag}$, so that 
     \begin{equation}
         \sum_{k} \Pi_QE^{\dag}R_{k}^{\dag}R_{k}F\Pi_Q = \Pi_QE^{\dag}F\Pi_Q = \sum_{k,i,j} c(s) U(E)^{*}_{ki}\Pi_QE_{i}^{\dag}U(F)_{kj}E_{j}\Pi_Q.
     \end{equation}
    Now, we can perform the sum over $k$, and so we get $\Pi_QE^{\dag}F\Pi_Q = c(s)\alpha(E,F)\Pi_Q$, with $\alpha(E,F) = \sum_{k,i,j} U^{*}(E)_{ki}\Pi_Q E_{i}^{\dag}U(F)_{kj}E_{j}$. Note that since the recovery map acts through at most an $L-$dimensional subspace, we can take there to be at most $L$ operators $R_{k}$.

    (Sufficient). Now we consider the reverse direction, and show that there exists a CPTP recovery map $\mathcal{R}$ that satisfies our conditions. 
    Let us define $Y_{k}(E) = \sum_{i}U(E)_{ki}E_{i}$ for some unitary $U(E)$. Then we rewrite $\Pi_Q E^{\dag}F\Pi_Q = \sum_{k} c(s)\Pi_QY_{k}(E)^{\dag}Y_{k}(F)\Pi_Q$. We add an auxiliary Hilbert space as the span of $\{ \ket{\ell}, \ell \in [L]\}$, and define $V(s)$ as, 
    \begin{equation}
        V(s)E\Pi_Q\ket{\psi} = \sqrt{c(s)}\sum_{k \in [L]} Y_{k}(E)\Pi_Q \ket{\psi} \otimes \ket{k}
    \end{equation}
    where $V(s)$ is a map \textit{restricted} to the syndrome $s$ subspace, which is the subspace of states having the syndrome value $s$. Let $\Pi(s)$ be the projector onto the subspace with syndrome $s$. First we note that for any $\ket{\alpha}$ in the syndrome $s$ subspace, 
    we can write $\ket{\alpha} = E\Pi_Q\ket{\psi}$ for some $\ket{\psi}$, and $E \in \mathcal{E}$ with syndrome $s$. We can then use this fact to check that $V(s)$ is an isometry on the subspace corresponding to syndrome $s$, since 
    \begin{equation}
        \braket{V(s)E\Pi_Q(\psi)}{V(s)F\Pi_Q(\phi)} = \bra{\psi}c(s) \sum_{k}\Pi_QY_{k}(E)^{\dag}Y_{k}(F)\Pi_Q\ket{\phi} = \bra{\psi}\Pi_QE^{\dag}F\Pi_Q\ket{\phi}
    \end{equation}
    by construction, and $V(s)\Pi(s)$ is partial isometry on the whole space. Then $V = \sum_{s} V(s)\Pi(s)$ is a full isometry.

    The claim is then that the recovery operator is given by $ \mathcal{R}(\cdot) = \mathrm{Tr}_{\mathrm{Aux}}(V(\cdot)V^{\dag})$, where we are tracing out the auxiliary qubits. Then we have 
    \begin{equation}
        \begin{split}
            \mathcal{R}(E(\Pi_Q\rho \Pi_Q)E^{\dag}) = \mathrm{Tr}_{\mathrm{Aux}}\left(V(E(\Pi_Q\rho \Pi_Q)E^{\dag})V^{\dag}\right) &= c(s) \mathrm{Tr}_{\mathrm{Aux}} \left( \sum_{k,k' \in L} Y_{k}(E)^{\dag}\Pi_Q\ketbra{\psi}{\psi}\Pi_Q Y_{k'}(E)  \otimes \ketbra{k}{k'} \right)\\
            &= c(s) \sum_{k} Y_{k}(E)^{\dag}\rho Y_{k}(E)\\
            &\stackrel{(a)}{=}c(s)\sum_{k, i, j} U(E)_{ki} E_i \rho U(E)_{kj}^{\star} E_j^{\dagger} \\
            &\stackrel{(b)}{=}c(s) \sum_{k, i, j} U(E)_{ki}U(E)^{\dagger}_{jk} E_i \rho  E_j^{\dagger} \\
            &\stackrel{(c)}{=}c(s) \sum_{i, j} (U^{\dagger} U)_{ji} E_i \rho E_j^{\dagger} \\
            &= c(s)\sum_{i} E_i \rho E_i^{\dagger}
        \end{split}
    \end{equation}
    where $(a)$ follows by using the definition of $Y_k(E)$, $(b)$ follows by rearrangement of terms and $(c)$ works out by doing the sum over $k$.

\end{proof}

\section{Detailed computation of Haar integrals}

\subsection{Probability calculations}
\label{appendix_wco}
\begin{theorem}\label{thm:prob_calcs_appendix_1}
    Let $P_1$ be, 
    \begin{equation}
        P_{1}(U)=1 -
        \mathrm{Tr} \bigg[(I_n\otimes \ketbra{0}{0}) U^\dagger  \bigg(E_1^s E_j^s  U\ketbra{\phi}{\phi} \otimes \ketbra{0}{0}^{\otimes m} 
    U^{\dag} E_j^s  E_1^s \bigg)  U \bigg]
\end{equation}
with $E_1^s$ and $E_j^s$ are the first element of the list and the true error respectively, and $E_1^s\neq E_j^s$. Then, 
    \begin{equation}
        \int_U P_1(U) dU \sim 1-2^{-m}.
    \end{equation}
\end{theorem}
\begin{proof}

We begin with the integral for the probability of projection onto the $\ketbra{0}{0}^{\otimes m}$ state when the wrong correction operator is applied. For simplicity, we denote the $m$ ancillary qubits as simply in state $|0\rangle \langle 0|$. Furthermore, denote $I_n \otimes \ketbra{0}{0} = \Pi$ and $P_0(U) = 1- P_1(U)$. Finally, we denote $\ket{\phi} \otimes \ket{0} = \ket{\phi,{0}}$, and $\ketbra{\phi,{0}}{\phi,{0}} = \rho_{0}$. Then, the Haar average of $P_0(U)$ is

\begin{equation}
   \Eset{U} {P}_0(U) =  \Eset{U}\bigg[\mathrm{Tr} \bigg(\Pi U^\dagger  E_1^s E_j^s  U \rho_{0}
    U^{\dag}  E_j^s E_1^s U \bigg) \bigg].
\end{equation}

Let us now write the shorthand $K^{\dagger}$ for $E_1^sE_j^s$ to get 

\begin{equation}
   \Eset{U} {P}_0(U) =  \bigg[\mathrm{Tr} \bigg(\Pi U^\dagger K^{\dagger} U \rho_{0} U^{\dag} K U \bigg) \bigg].
\end{equation}

We can use the cyclic property of the trace to write this as 

\begin{equation}
    \Eset{U} {P}_0(U) =   \Eset{U}\bigg[\mathrm{Tr} \bigg( U^{\dag} K U \Pi U^\dagger K^{\dagger} U \rho_{0} \bigg) \bigg].
\end{equation}

Now we use $AB = \mathrm{Tr}_{2}(A \otimes B \mathbb{F})$, where $\mathbb{F}$ is the SWAP operator on the two subsystems, $\mathrm{Tr_{2}}$ is the partial trace over the second subsystem, and also the fact that $\mathrm{Tr}(\mathrm{Tr}_{2} (\cdot)) = \mathrm{Tr}(\cdot)$ to get

\begin{equation}
    \Eset{U} {P}_0(U) = \Eset{U}\bigg[\mathrm{Tr}\bigg( \big(U^{\dag}KU\big) \otimes \big(\Pi U^{\dag}K^{\dag}U \rho_{0} \big) \mathbb{F} \bigg)\bigg],
\end{equation}
which we can separate out to write above expectation in the form $\mathbb{E}[U^{\otimes 2}(\cdot)U^{\dag \otimes 2} ]$ as 

\begin{equation}
    \Eset{U} {P}_0(U) = \Eset{U}\bigg[\mathrm{Tr}\bigg( \big( I_{n+m} \otimes \Pi \big) U^{\dag \otimes 2} \big( K \otimes K^{\dag} \big)
    U^{\otimes 2} \big( I_{n+m} \otimes \rho_{0} \big) \mathbb{F} \bigg)\bigg].
\end{equation}

We can now use the standard results in computing Haar expectations~\cite{mele2024introduction} for the operator $\mathcal{O} = K \otimes K^{\dagger}$

\begin{equation}
    \Eset{U}\bigg[ U^{\dag \otimes 2} \mathcal{O} U^{\otimes 2}\bigg] = c_{I, \mathcal{O}} I +c_{\mathbb{F}, \mathcal{O}} \mathbb{F},
\end{equation}

where the coefficients $c_{I, \mathcal{O}}$ and $c_{\mathbb{F}, \mathcal{O}}$ can be found explicitly given a full description of $\mathcal{O}$ and $d = 2^{m+n}$ as

\begin{equation}
    c_{I, \mathcal{O}} = \frac{\mathrm{Tr}(\mathcal{O}) - \frac{1}{d}\mathrm{Tr}(\mathcal{O}\mathbb{F})}{d^{2}-1}
\end{equation}

\begin{equation}
    c_{\mathbb{F}, \mathcal{O}} = \frac{\mathrm{Tr}(\mathbb{F}\mathcal{O}) - \frac{1}{d}\mathrm{Tr}(\mathcal{O})}{d^{2}-1}.
\end{equation}

In our case, we have $\mathcal{O} =  K \otimes K^{\dag}$, and the  expressions carry over directly into the integration, yielding

\begin{equation}
    \Eset{U} {P}_0(U) = c_{I} \bigg[\mathrm{Tr}\bigg( \big( I_{n+m} \otimes \Pi \big) \big( I_{n+m} \otimes \rho_{0} \big) \mathbb{F} \bigg)\bigg] 
    + c_{\mathbb{\mathbb{F}}} \bigg[\mathrm{Tr}\bigg( \big( I_{n+m} \otimes \Pi \big) \mathbb{F} \big( I_{n+m} \otimes \rho_{0} \big) \mathbb{F} \bigg)\bigg],
\end{equation}
where we have suppressed the operator in the subscript of the coefficients, that is, $c_{I} = c_{I, K \otimes K^{\dag}}$, and similarly for $c_{\mathbb{F}}$. Using the fact that $\mathrm{Tr}(A \otimes B \mathbb{F}) = \Tr(AB )$, the first of the two terms solves immediately, since $c_{I} \big[\mathrm{Tr}\left( \big( I_{n+m} \otimes \Pi \big) \big( I_{n+m} \otimes \rho_{0} \big) \mathbb{F} \right)\big] = c_{I}\mathrm{Tr}(\Pi \rho_{0})$, and $\mathrm{Tr}(\Pi \rho_{0}) = 1$. Let us call the second term $T_{2}$, and work it out, starting from

\begin{equation}
    T_{2} = c_{\mathbb{F}}\mathrm{Tr}\bigg( \big( I_{n+m} \otimes \Pi\big)\mathbb{F}\big( I_{n+m} \otimes \rho_{0}\big) \mathbb{F}\bigg).
\end{equation}

We can directly simplify $T_2$ by using the action of $\mathbb{F}$ operator on $I_{n+m} \otimes \Pi$, which interchanges the matrices. We then multiply it with $I_{n+m} \otimes \rho_0$ and take the trace to get $2^n$ directly. Formally,

\begin{equation}
    \mathbb{F}\left(I_{n+m} \otimes \Pi\right) \mathbb{F} = \Pi \otimes I_{n+m}.
\end{equation}

Thus,

\begin{equation}
    T_2 = c_{\mathbb{F}}\Tr \left((\Pi \otimes I_{n+m})(I_{n+m} \otimes \rho_{0})\right) = c_{\mathbb{F}}\Tr \Pi = c_{\mathbb{F}} 2^n.
\end{equation}

Altogether, we have

\begin{equation}
    \Eset{U} {P}_0(U) = c_{I} + 2^{n}c_{\mathbb{F}}.
\end{equation}

Now we explicitly evaluate the coefficients $c_{I}$ in terms of $K$:

\begin{equation}
    c_{I} = \frac{|\mathrm{Tr}(K)|^{2} - 1}{d^{2}-1}
\end{equation}

\begin{equation}
    c_{\mathbb{F}} = \frac{(d - \frac{|\mathrm{Tr}(K)|^{2}}{d})}{d^{2}-1},
\end{equation}
and as $K^{\dag} = E_1^s E_j^s$ with $E_1^s$ and $E_j^s$ as Pauli matrices, $K$ is another non-identity Pauli matrix up to some phase, thus $\Tr K = 0$. So, we can write

\begin{equation}
    \Eset{U} {P}_0(U) = \frac{-1}{d^2 - 1} + \frac{2^n d}{d^2 - 1} \leq \frac{2^n d}{d^2 - 1} \sim 2^{-m}
\end{equation}
and so, $\Eset{U} {P}_1(U) = 1 - \Eset{U} {P}_0(U) = 1 - O(2^{-m})$, as required.
\end{proof}

Now, we can prove the Proposition~\ref{thm:prob_calcs_main_text}.

\begin{proposition}[Restatement of \cref{thm:prob_calcs_main_text}]\label{thm:prob_calcs_appendix}
    Let $P_{\lambda, 1}$ be
    \begin{equation}
        P_{\lambda, 1} = 1 -  \,  \mathrm{Tr} \bigg[(I_n\otimes \ketbra{0}{0}) U_{\lambda}^\dagger  \bigg(E^{s}_1 E_{j}^{s}  U_{\lambda}\ketbra{\phi,0}{\phi,0} \\
    U_{\lambda}^{\dag} E_{1}^{s \dag}E_j^{s \dag}\bigg)  U_{\lambda} \bigg],
    \end{equation}
    then
    \begin{equation}
        P_{\lambda, 1} \geq 1-\delta \sim 1-\negl(n)
    \end{equation}
    with probability $1 - {O}(2^{-m}/\delta) \sim 1 - \negl(n)$ over keys for $\delta = \Theta(2^{-m+\poly\log n})$ and $m = \omega(\log n)$,
    under the assumptions that the PRU forms an exact $2$-design. 
\end{proposition}
\begin{proof}
    Given that the PRU forms an exact $2$-design, we have
    \begin{equation}
        \Eset{\lambda}\, P_{\lambda, 1} = \Eset{U} P_1(U)
    \end{equation}
    and thus, from Theorem~\ref{thm:prob_calcs_appendix_1}, we get that $\Eset{\lambda} P_{\lambda, 1} = 1 - {O}(2^{-m})$. Now, from Markov's inequality, 
    \begin{equation}
    \Pr_{\lambda}[1 - P_{\lambda, 1} > \delta ] \leq \frac{1 - \mathbb{E}_\lambda {P}_{\lambda, 1}}{\delta} = \frac{{O}(2^{-m})}{\delta}.
\end{equation}
    Thus, the probability $P_{\lambda, 1} \geq 1 - \delta$ with probability $1 - O\left(\frac{2^{-m}}{\delta}\right)$. Putting $\delta = \Theta(2^{-m + \poly \log (n)})$ and $m = \omega(\log n)$ then gives the required asymptotics. 
\end{proof}

\subsection{Fidelity calculations}
\label{appendix_fidelity_calcs}

\begin{theorem}[Restatement of \Cref{prop:fidelity_calcs_main_text}]
    The fidelity $\mathcal{F}(\rho_{\lambda, 1}, \rho_{\lambda, 0})$ given as
    \begin{equation}
 \mathcal{F}(\rho_{\lambda, 1}, \rho_{\lambda, 0}) =  \bigg(\frac{1}{P_{\lambda, 1}}\Tr  \Big[\rho_{\lambda, 0}E_1^{s}   U_{\lambda}(I-\Pi)U_{\lambda}^{\dagger} E^{s}_1 \rho_{\lambda, 0} E_1^{s}  U_{\lambda} (I-\Pi) U_{\lambda}^{\dagger}  E^{s}_1\Big] \bigg).
\end{equation}
with $\rho_{\lambda, 0} = E_j^s U_{\lambda} |\phi, 0 \rangle \langle \phi, 0| U_{\lambda}^{\dag} E_j^s$ is the initial state, $\rho_{\lambda, 1} = \frac{1}{P_{\lambda, 1}}E^{s}_1 U_{\lambda} (I-\Pi) U_{\lambda}^{\dagger}E^{s}_1 \rho_{\lambda, 0} E^{s}_1 U_{\lambda} (I-\Pi) U_{\lambda}^{\dagger} E^{s}_1$ with $E_1^s$ and $E_j^s$ are list errors such that $E_1^s \neq E_j^s$, then
    \begin{equation}
        \mathcal{F}(\rho_{\lambda, 1}, \rho_{\lambda, 0}) \geq 1 - \delta \sim 1 - \negl(n)
    \end{equation}
   with probability $1 - O\left(\frac{2^{-m} + \epsilon}{\delta}\right) \sim 1 - \negl(n)$ over keys assuming that $\{U_{\lambda}\}_{\lambda}$ forms an approximate 4-design with relative error $\epsilon = 2^{-m}$ and an exact $2$-design, for $\delta = \Theta(2^{-m + \text{polylog}(n)})$ and $m = \omega(\log n)$.
\end{theorem}
\begin{proof}
    It is important to check whether our protocol works with only a single state, i.e. we can re-use a single copy of the state with a list error to uniquely decode it. Our protocol involves: first applying a list correction operator (which can be a wrong one), then applying the inverse of PRU with the same key, and then projecting on $\Pi = I_n \otimes |0\rangle \langle0|^{\otimes m}$. In case of a wrong correction operator, we know that the projection fails with high probability. Thus, in this case, we need to check that after projecting onto $I - \Pi$, we can go back to the initial list-error state upto high fidelity by using the inverse of the wrong correction operator selected. The un-normalised state after projection is formally given by
\begin{equation}
    \sigma = (I-\Pi)U_{\lambda}^{\dagger} E_1^s E_j^s (U_{\lambda}\rho_0 U_{\lambda}^{\dagger}) E^s_j E_1^{s}  U_{\lambda} (I-\Pi)
\end{equation}
where $\rho_0 = |\phi, 0\rangle\langle\phi, 0|$, $E_j^s$ represent the true error, $E_1^s$ represent the first correction operator applied from the list with syndrome $s$, $U_{\lambda}$ represent the pseudorandom unitary with key $\lambda$.
To go back to the initial state, we first apply the $U_{\lambda}$, and then the inverse of the wrong correction operator to get the un-normalised state,
\begin{equation}
    \tilde{\rho}_{\lambda, 1} = E_1^{s}  U_{\lambda}(I-\Pi)U_{\lambda}^{\dagger} E_1^s E_j^s (U_{\lambda}\rho_0 U_{\lambda}^{\dagger})  E^{s}_j E_1^{s}  U_{\lambda} (I-\Pi) U_{\lambda}^{\dagger}  E_1^s .
\end{equation}
Then, the un-normalised fidelity with true error state $\rho_{\lambda, 0} = E_j^s (U_{\lambda}\rho_0 U_{\lambda}^{\dagger})  E^{s}_j$ is
\begin{equation}
    \mathcal{F}(\tilde{\rho}_{\lambda, 1}, \rho_{\lambda, 0}) = \Tr \left[E_j^s (U_{\lambda}\rho_0 U_{\lambda}^{\dagger}) E^{s}_j E_1^{s} U_{\lambda}(I-\Pi)U_{\lambda}^{\dagger} E_1^s E_j^s (U_{\lambda}\rho_0 U_{\lambda}^{\dagger})  E_j^{s} E_1^{s}  U_{\lambda} (I-\Pi) U_{\lambda}^{\dagger}  E_1^s\right].
\end{equation}
By rearranging and defining $T = E_1^s E_j^s $, we can write above as
\begin{equation}
    \mathcal{F}(\tilde{\rho}_{\lambda, 1}, \rho_{\lambda, 0}) = \Tr \left[ \left(\rho_0 \otimes (I-\Pi) \otimes \rho_0 \otimes (I-\Pi)\right) U_{\lambda}^{\dagger, \otimes 4} \left(T^{\dagger}\otimes T\otimes T^{\dagger} \otimes T\right) U_{\lambda}^{\otimes 4} G_{\gamma}\right]
\end{equation}
where $G_{\gamma}$ is the permutation unitary corresponding to the cyclic permutation $\gamma = (1234)$ and we used the fact that $\Tr[A \otimes B \otimes C \otimes D G_{\gamma}] = \Tr[ABCD]$ for any cyclic permutation $\gamma$ on $S_4$ group. The average fidelity is then given by
\begin{equation}
    \mathbb{E}_{\lambda} \mathcal{F}(\tilde{\rho}_{\lambda, 1}, \rho_{\lambda, 0}) = \Tr \left[ \left(\rho_0 \otimes (I-\Pi) \otimes \rho_0 \otimes (I-\Pi)\right) \mathbb{E}_{\lambda}  \left[U_{\lambda}^{\dagger, \otimes 4} \left(T^{\dagger}\otimes T\otimes T^{\dagger} \otimes T\right) U_{\lambda}^{\otimes 4}\right] G_{\gamma}\right]
\end{equation}

To calculate the expectation on keys, we first calculate the average on Haar random unitaries:

\begin{equation}\label{eq:expected_fidelity}
    \Eset{U} \mathcal{F}(\tilde{\rho}_{U, 1}, \rho_{U, 0}) = \Tr \left[ \left(\rho_0 \otimes (I-\Pi) \otimes \rho_0 \otimes (I-\Pi)\right) \Eset{U}  \left[U^{\dagger, \otimes 4} \left(T^{\dagger}\otimes T\otimes T^{\dagger} \otimes T\right) U^{\otimes 4}\right] G_{\gamma}\right].
\end{equation}
Now, we can write the Haar integral using Weingarten calculus~\cite{mele2024introduction} as:
\begin{equation}\label{eq:haar_integral_fourth_order}
    \Eset{U} \left[U^{\dagger, \otimes 4} \mathcal{O} U^{\otimes 4}\right ] = \sum_{\pi, \sigma \in S_4} Wg(\pi^{-1}\sigma, d) \Tr(V_d(\sigma)^{\dagger} \mathcal{O}) V_d(\pi)
\end{equation}
where $V_d(\pi)$ is the unitary representation for the permutation $\pi$ over $d$-dimensional Hilbert space and $\mathcal{O} =  \left(T^{\dagger}\otimes T\otimes T^{\dagger} \otimes T\right)$. Now, using~\eqref{eq:expected_fidelity} and~\eqref{eq:haar_integral_fourth_order}, we get:
\begin{equation}
    \Eset{U} \mathcal{F}(\tilde{\rho}_{U, 1}, \rho_{U, 0}) = \sum_{\pi, \sigma \in S_4} Wg(\pi^{-1}\sigma, d) \underbrace{\Tr(V_d(\sigma)^{\dagger} \mathcal{O})}_{\beta_{\sigma}} \underbrace{\Tr \left[ \left(\rho_0 \otimes (I-\Pi) \otimes \rho_0 \otimes (I-\Pi)\right)  V_d(\pi) G_{\gamma}\right]}_{\alpha}.
\end{equation}
Now, note that $\rho_0(I-\Pi) = 0$, thus only those permutation $\pi$ survive, for which $\alpha$ does not lead to $\rho_0(I-\Pi)$ terms. It is easy to see that there will only be four such permutations, $\pi = (4321), (12)(34), (14)(23), (1234)$, corresponding to $\alpha = (\Tr \rho)^2 (\Tr(I-\Pi))^2, \Tr \rho^2  (\Tr(I-\Pi))^2, (\Tr \rho)^2 \Tr(I-\Pi)^2, \Tr \rho^2 \Tr(I-\Pi)^2$ respectively.

We will now abuse the notation a bit by denoting the permutation $\pi$ by its cycle structure form. Then, we have,
\begin{equation}\label{eq:expression_expected_fidelity}
    \begin{split}
        \Eset{U} \mathcal{F}(\tilde{\rho}_{U, 1}, \rho_{U, 0}) =& \sum_{\sigma} \Big[ Wg\left((4321)^{-1}\sigma, d\right) \beta_{\sigma} (\Tr \rho)^2 (\Tr(I-\Pi))^2 + Wg\left(((12)(34))^{-1}\sigma, d\right) \beta_{\sigma} \Tr \rho^2  (\Tr(I-\Pi))^2 + \\
        &Wg\left(((14)(23))^{-1}\sigma, d\right) \beta_{\sigma} (\Tr \rho)^2 \Tr(I-\Pi)^2 + Wg\left((1234)^{-1}\sigma, d\right) \beta_{\sigma} \Tr \rho^2 \Tr(I-\Pi)^2 \Big].
    \end{split}
\end{equation}
Note that for $\mathcal{O} = T^{\dag} \otimes T \otimes T^{\dag} \otimes T$, we get  $\beta_{\sigma} = \Tr \mathcal{O} V_d(\sigma)^{\dagger} = {O}(d^{\# \mathrm{cycles}(\sigma)}))$ as we show in Lemma~\ref{lemma:assumption_O}.

We use Weingarten coefficients that can be computed using the \texttt{RTNI} package~\cite{Fukuda_2019} and for the order four Haar integral are mentioned below,
\begin{equation}\label{eqn:Wg_coeffs}
    \begin{aligned}
        & Wg\left([1, 1, 1, 1], d\right)=\frac{d^4-8 d^2+6}{d^2\left(d^2-1\right)\left(d^2-4\right)\left(d^2-9\right)} \\
        & Wg\left([2, 1, 1], d\right)=\frac{-1}{d\left(d^2-1\right)\left(d^2-9\right)} \\
        & Wg\left([2, 2], d\right)=\frac{d^2+6}{d^2\left(d^2-1\right)\left(d^2-4\right)\left(d^2-9\right)} \\
        & Wg([3, 1], d)=\frac{2 d^2-3}{d^2\left(d^2-1\right)\left(d^2-4\right)\left(d^2-9\right)} \\
        & Wg([4], d)=\frac{-5}{d\left(d^2-1\right)\left(d^2-4\right)\left(d^2-9\right)}  .
    \end{aligned}
\end{equation}
Using Lemma~\ref{lemma:assumption_O} and Weingarten coefficients~\eqref{eqn:Wg_coeffs}, the leading order terms in~\eqref{eq:expression_expected_fidelity} can be deduced with a bit of work to be,
\begin{equation}
    \begin{split}
        \Eset{U} \mathcal{F}(\tilde{\rho}_{U, 1}, \rho_{U, 0}) &= Wg([1, 1, 1, 1], d)  \Tr T^{\dagger} T \Tr T T^{\dagger} [\Tr(I-\Pi)]^2 \\
        &= \frac{d^4-8 d^2+6}{d^2\left(d^2-1\right)\left(d^2-4\right)\left(d^2-9\right)} d^2  [\Tr(I-\Pi)]^2 \\
        &= \frac{d^6}{d^2\left(d^2-1\right)\left(d^2-4\right)\left(d^2-9\right)} \left(d - \frac{d}{2^m}\right)^2 + {O}\left(\frac{1}{d^2}\right) \\
        &= \frac{d^6}{\left(d^2-1\right)\left(d^2-4\right)\left(d^2-9\right)}\left(1 - {O}\left(\frac{1}{2^{m}} \right) \right) + {O} \left(\frac{1}{d^2}\right) \\
        &= \left(1 + {O}\left(\frac{1}{d^2}\right)\right)\left(1 - {O}\left(\frac{1}{2^{m}} \right) \right) \\
        &= 1 - {O}\left(\frac{1}{2^{m}}\right) + {O}\left(\frac{1}{d^2}\right).
    \end{split}
\end{equation}

Now, using the assumption that PRU is a relative error $\epsilon$-approximate $4$-design, we get that

\begin{equation}
    (1-\epsilon) \Eset{U}[U^{\dagger, \otimes 4}\mathcal{O} U^{\otimes 4}] \preceq \mathbb{E}_\lambda[U_{\lambda}^{\dagger, \otimes 4}\mathcal{O} U_{\lambda}^{\otimes 4}] \preceq (1+\epsilon) \Eset{U}[U^{\dagger, \otimes 4} \mathcal{O} U^{\otimes 4}].
\end{equation}

Then, denoting $\zeta = \rho \otimes (I-\Pi) \otimes \rho \otimes (I-\Pi)$, $\mathcal{F}(\tilde{\rho}_{U, 1}, \rho_{U, 0}) = \tilde{\mathcal{F}}_{U}$ and $\mathcal{F}(\tilde{\rho}_{\lambda, 1}, \rho_{\lambda, 0}) = \tilde{\mathcal{F}}_{\lambda}$
\begin{equation}
    \begin{split}
        \left|\mathbb{E}_{\lambda} \tilde{\mathcal{F}}_\lambda - \mathbb{E}_{U} \tilde{\mathcal{F}}_U\right| &= \left|\Tr\left[\zeta\left(\mathbb{E}_\lambda[U_{\lambda}^{\dagger, \otimes 4} \mathcal{O} U_{\lambda}^{\otimes 4}] - \Eset{U}[U^{\dagger, \otimes 4} \mathcal{O} U^{\otimes 4}] \right)G_{\gamma}\right]\right| \\
        &\leq \epsilon \left|\Tr\left[\zeta\left(\Eset{U}[U^{\dagger, \otimes 4} \mathcal{O} U^{\otimes 4}] \right)G_{\gamma}\right]\right| \\
        &= \epsilon \Eset{U} \tilde{\mathcal{F}}_U = \epsilon - {O}\left(\frac{\epsilon}{2^{m}}\right).
    \end{split} 
\end{equation}

Thus, $\mathbb{E}_\lambda \tilde{\mathcal{F}}_\lambda \geq 1- \epsilon - {O}\left(\frac{1}{2^{m}}\right)$. Now, note that from Markov's inequality,
\begin{equation}
    \Pr_{\lambda}[1 - \tilde{\mathcal{F}}_\lambda > \delta ] \leq \frac{1 - \mathbb{E}_\lambda \tilde{\mathcal{F}}_\lambda}{\delta} = \frac{\epsilon + {O}(2^{-m})}{\delta} .
\end{equation}
Thus, the un-normalised fidelity $\tilde{\mathcal{F}}_\lambda = \mathcal{F}(\tilde{\rho}_{\lambda, 1}, \rho_{\lambda, 0}) \geq 1 - \delta$ with probability $1 - O\left(\frac{\epsilon + 2^{-m}}{\delta}\right)$. Then putting $\delta = \Theta(2^{-m + \text{polylog}(n)})$, $m = \omega(\log n)$ and $\epsilon = O(2^{-m})$ will give $\tilde{\mathcal{F}}_\lambda \geq 1 - \Theta(2^{-m + \text{polylog}(n)}) \sim 1 - \negl(n)$ with probability $1 - O\left(\frac{2^{-m}}{\delta}\right) \sim 1 - {O}(2^{-\text{polylog}(n)})\sim 1 - \negl(n)$ over keys. Finally, note that the normalised fidelity,
\begin{equation}
    \mathcal{F}(\rho_{\lambda, 1}, \rho_{\lambda, 0}) = \frac{\mathcal{F}(\tilde{\rho}_{\lambda, 1}, \rho_{\lambda, 0})}{P_{\lambda, 1}} \geq \mathcal{F}(\tilde{\rho}_{\lambda, 1}, \rho_{\lambda, 0}).
\end{equation}
Thus, we get the required asymptotics for normalised fidelity.

\end{proof}

\begin{lemma}\label{lemma:assumption_O}
    For $O =  \left(T^{\dagger}\otimes T\otimes T^{\dagger} \otimes T\right)$, with $T$ as defined above, we have that:
    \begin{equation}
        \beta_{\sigma} = \Tr O V_d(\sigma)^{\dagger} = {O}(d^{\# \mathrm{cycles}(\sigma)})).
    \end{equation}
\end{lemma}
\begin{proof}
   We note that $T = E_1^sE_j^s$ is a Pauli matrix upto some phase, and hence: 
   \begin{equation}
       \Tr T = 0 = \Tr T^{\dagger}.
   \end{equation}
    Thus, the only surviving terms in $\beta_{\sigma}$ will be the ones for which $\sigma$ either has only a single cycle or it has $(2, 2)$ cyclic structure, as otherwise, there would be single trace terms as factors which will make the term 0. Now, we look at these surviving terms turn-wise:
    \begin{enumerate}
        \item \textbf{Single cycle terms:}
        \begin{equation}
            \begin{split}
                 \Tr T^{\dagger} T TT^{\dagger} = \Tr I = d
            \end{split}
        \end{equation}
        where we just used the fact that $T^{\dagger}T = E_j^{s}E_1^{s} E_1^s E_j^s = I$.
        \item \textbf{Terms with cycle structure (2, 2):}
        \begin{equation}
             \Tr T^{\dagger} T \Tr T^{\dagger} T = {O}(d^2),
        \end{equation}
        which is obvious from the above calculations.
        Now, 
        \begin{equation}
            \begin{split}
               \Tr T^{\dagger} T^{\dagger} \Tr T T &= \left|\Tr T^{\dagger} T^{\dagger}\right|^2 
                = d^2,
            \end{split}
        \end{equation}
        where we used the fact that $T^{\dagger, 2}$ is $I$ up to $\pm 1$. The rest of the terms of the structure $(2, 2)$ are equivalent to one of the above terms. 
    \end{enumerate}
\end{proof}

\section{Unique quantum list decoding}\label{appendix_unique_decoding_proof}

In this section, we prove our theorem on our list decoding scheme and its corollary. 
Here, we denote the logical state as $\rho = |\phi\rangle \langle \phi|$.

\begin{theorem}[Restatement of \Cref{thm:almost_pure_state}]
\renewcommand{\thetheorem}{2}
Let $Q_L$ be an $\dsl n^{\prime}, n+m\dsr$ stabilizer code which is also an $L$-list code for an error set $\mathcal{E}$. Then the Algorithm~\ref{alg:unambiguous_qld} returns the correct logical state with fidelity at least $1 - \Theta(L 2^{-m/2 + \text{polylog}(n)}) \sim 1-\negl(n)$ with success probability $1 - \negl(n)$ over the randomness in the algorithm and with probability $1 - O\left(\frac{2^{-m}}{\delta}\right) \sim 1 - \negl(n)$ over keys, for $\delta = \Theta(2^{-m+\poly\log (n)})$, $m = \omega(\log n)$ and $L = {O}(\poly(n))$
assuming that $\{U_{\lambda}\}_{\lambda}$ is both an approximate $4$-design with relative error $\epsilon = 2^{-m}$ and an exact $2$-design.
\end{theorem}
\begin{proof}
    Assume that given the initial state, our algorithm halts after $t_f+1$ steps. We analyze the case when the algorithm halts correctly, i.e. when the correct list error is found. Later, we show that this happens with very high probability. Now, we denote the initial state as $\rho_{\lambda, 0} = E_j U_{\lambda}\rho \otimes |0\rangle \langle 0|^{\otimes m} U_{\lambda}^{\dagger}E_j$ for some $1 \leq j \leq L$ and define $A_{\lambda, i} = E^{s}_iU_{\lambda}(I-\Pi)U_{\lambda}^{\dagger}E^{s}_i$. It is easy to check that the matrix $A_{\lambda, i}$ is hermitian.\\
    Now, note that first step of Algorithm~\ref{alg:unambiguous_qld} takes the initial state $\rho_{\lambda, 0}$ to $\rho_{\lambda, 1} = \frac{1}{P_{\lambda, 1}}A_{\lambda, 1}\rho_{\lambda, 0}A_{\lambda, 1}^{\dagger}$, with $P_{\lambda, 1} = \tr A_{\lambda, 1}\rho_{\lambda, 0}A_{\lambda, 1}^{\dagger} $ which also gives the probability of witnessing $(I-\Pi)$ measurement outcome on $U_{\lambda}^{\dagger}E^{s}_i\rho_{\lambda, 0}E^{s}_iU_{\lambda}$ and in general $i$-th (for $i\leq t_f$) step takes $\rho_{\lambda, i-1}$ to $\rho_{\lambda, i} = \frac{1}{P_{\lambda, i}}A_{\lambda, i}\rho_{\lambda, i-1} A_{\lambda, i}$ where $P_{\lambda, i} = \tr A_{\lambda, i} \rho_{\lambda, i-1} A_{\lambda, i}$ and it represents the probability of witnessing the outcome $(I-\Pi)$ on the state $U_{\lambda}^{\dagger} E^{s}_i \rho_{\lambda, i-1} E^{s}_i U_{\lambda}$. And the $t_f+1$-th step will take the state $\rho_{\lambda, t_f}$ to $\rho_{\lambda, t_f+1} = \frac{1}{P_{\lambda, t_f+1}} \Pi U_{\lambda}^{\dagger} E^{s}_{t_f+1} \rho_{\lambda, t_f} E^{s}_{t_f +1} U_{\lambda} \Pi$ with $P_{\lambda, t_f+1}$ represents the probability of witnessing the outcome $\Pi$ on the state $U_{\lambda}^{\dagger}E^{s}_{t_f+1}\rho_{\lambda, t_f}E^{s}_{t_f+1}U_{\lambda}$. So, a successful run of the algorithm looks like, 
    \begin{equation}
        \boxed{\rho_{\lambda,0}} \Rightarrow \boxed{\rho_{\lambda, 1}} \ldots \ldots \Rightarrow \boxed{\rho_{\lambda, t_f}} \Rightarrow \boxed{\rho_{\lambda, t_f +1}}
    \end{equation}
     Now, our goal is to calculate $\mathcal{F}(\rho_{\lambda, t_f +1}, \rho \otimes |0\rangle \langle0|^{\otimes m})$, which is the fidelity between the state at the end of Algorithm~\ref{alg:unambiguous_qld} and the logical state. Then, we see that, 
     \begin{equation}
    \begin{split}
        \mathcal{F}(\rho_{\lambda, t_f}, \rho\otimes  |0\rangle \langle 0|^{\otimes m} ) 
        &\stackrel{(a)}{\geq} \tr ( \Pi U_{\lambda}^{\dagger} E^{s}_{t_f+1} \rho_{\lambda, t_f} E^{s}_{t_f+1} U_{\lambda} \Pi \left(\rho \otimes |0\rangle \langle 0|^{\otimes m}\right)) \\
        & \stackrel{(b)}{=} \tr (\rho_{\lambda, t_f} E^{s}_{t_f+1} U_{\lambda} \left(\rho \otimes |0\rangle \langle 0|^{\otimes m}\right) U_{\lambda}^{\dagger} E^{s}_{t_f+1}) \\
        & \stackrel{(c)}{=} \tr(\rho_{\lambda, t_f} \rho_{\lambda, 0}) \\ 
        &= \mathcal{F}(\rho_{\lambda, t_f}, \rho_{\lambda, 0}), \\
    \end{split}
    \end{equation}
    where $(a)$ follows from lower bounding the normalised fidelity by the un-normalised fidelity, (b) follows from cyclic property of trace and noting that $\Pi(\rho \otimes |0\rangle \langle 0 |^{\otimes m})\Pi = \rho \otimes |0\rangle \langle 0|^{\otimes m}$ and $(c)$ follows from the assumption that $E^{s}_{t_f+1} = E$ (correct list error) and using the definition of $\rho_{\lambda, 0}$.\\
     So, to calculate the required fidelity between the end state and the logical state, it suffices to lower bound the fidelity between the state after $t_f$ steps and the state $\rho_{\lambda, 0}$. Now, note that, for $i\leq t_f$, we can write $\rho_{\lambda, i} = \frac{1}{\tilde{P}_{\lambda, i}} \tilde{\rho}_{\lambda, i}$ with $\tilde{\rho}_{\lambda, i} = A_{\lambda, i}\ldots A_{\lambda, 1}\rho_{\lambda, 0}A_{\lambda, 1} \ldots A_{\lambda, i}$ and $\tilde{P}_{\lambda, i} = \tr (A_{\lambda, i}\ldots A_{\lambda, 1}\rho_{\lambda, 0}A_{\lambda, 1} \ldots A_{\lambda, i})$. We first calculate the un-normalised fidelity ${\mathcal F}(\tilde{\rho}_{\lambda, t_f}, \rho_{\lambda, 0})$ and then bound the normalised fidelity from below using the un-normalised fidelity. For $t_f = 2$, the un-normalised fidelity,
    \begin{equation}
    \begin{split}
        {\mathcal{F}}(\tilde{\rho}_{\lambda, 2} \rho_{\lambda, 0}) &= 
         \tr \left(A_{\lambda, 2}A_{\lambda, 1}\rho_{\lambda, 0}A_{\lambda, 1}A_{\lambda, 2}\rho_{\lambda, 0}\right)\\
        &= \tr \left(A_{\lambda, 1}\rho_{\lambda, 0}A_{\lambda, 1}A_{\lambda, 2}\rho_{\lambda, 0}A_{\lambda, 2}\right).
    \end{split}
    \end{equation}
    Now, we look at, 
    \begin{equation}\label{eq:thm2_main_step}
        \begin{split}
            |{\mathcal{F}}(\tilde{\rho}_{\lambda, 2}, \rho_{\lambda, 0}) - {\mathcal{F}}(\tilde{\rho}_{\lambda, 1}, \rho_{\lambda, 0})| 
            &= \left|\tr \left(A_{\lambda, 1}\rho_{\lambda, 0}A_{\lambda, 1}A_{\lambda, 2}\rho_{\lambda, 0}A_{\lambda, 2}\right) - \tr \left(A_{\lambda, 1}\rho_{\lambda, 0}A_{\lambda, 1}\rho_{\lambda, 0}\right)\right|\\
            &=\left|\tr\left(A_{\lambda, 1}\rho_{\lambda, 0}A_{\lambda, 1} (A_{\lambda, 2}\rho_{\lambda, 0}A_{\lambda, 2} - \rho_{\lambda, 0})\right)\right|\\
            &\stackrel{(a)}{\leq} \sqrt{\tr(A_{\lambda, 1}\rho_{\lambda, 0}A_{\lambda, 1})^2 \tr(A_{\lambda, 2}\rho_{\lambda, 0}A_{\lambda, 2} - \rho_{\lambda, 0})^2} \\
            &\stackrel{(b)}{\leq} \sqrt{\tr(A_{\lambda, 2}\rho_{\lambda, 0}A_{\lambda, 2} - \rho_{\lambda, 0})^2}\\
            & \stackrel{(c)}{\leq} \sqrt{2 - 2\tr(A_{\lambda, 2}\rho_{\lambda, 0}A_{\lambda, 2}\rho_{\lambda, 0})},
        \end{split}
    \end{equation}
    where $(a)$ follows from the Cauchy-Schwarz inequality, $(b)$ follows from noting that $\tr(A_{\lambda, 1}\rho_{\lambda, 0}A_{\lambda, 1})^2 \leq \frac{1}{P_{\lambda, 1}^2}\tr(A_{\lambda, 1}\rho_{\lambda, 0}A_{\lambda, 1})^2 \leq 1$ i.e. the purity of un-normalised state is bounded by 1 and $(c)$ also follows in similar fashion, by noting $\tr(A_{\lambda, 2}\rho_{\lambda, 0}A_{\lambda, 2})^2 \leq 1$. Now, note that the expression $\tr(A_{\lambda, 2}\rho_{\lambda, 0}A_{\lambda, 2}\rho_{\lambda, 0})$ is similar to fidelity expression except for $E^{s}_1$ list operators replaced by $E^{s}_2$ operators, and as the fidelity expression calculated in Proposition~\ref{prop:fidelity_calcs_main_text} is essentially independent of particular choice of list error, we can use the same fact here. In particular, we note that, with probability $1 - O\left(\frac{2^{-m}}{\delta}\right)$ over keys $\lambda$, we have $\tr(A_{\lambda, 2}\rho_{\lambda, 0}A_{\lambda, 2}\rho_{\lambda, 0}) \geq 1 - \delta$. Thus, we get,
    \begin{equation}
        \begin{split}
            |{\mathcal{F}}(\tilde{\rho}_{\lambda, 2}, \rho_{\lambda, 0}) - {\mathcal{F}}(\tilde{\rho}_{\lambda, 1}, \rho_{\lambda, 0})| \leq \sqrt{2\delta}.
        \end{split}
    \end{equation}
    So, ${\mathcal{F}}(\tilde{\rho}_{\lambda, 2}, \rho_{\lambda, 0}) \geq {\mathcal{F}}(\tilde{\rho}_{\lambda, 1}, \rho_{\lambda, 0}) - \sqrt{2\delta} = 1-\sqrt{2\delta} - \delta$. Now, we can repeat the same argument by writing $\mathcal{F}(\tilde{\rho}_{\lambda, 3}, \rho_{\lambda, 0}) = \tr\left(A_{\lambda, 2}A_{\lambda, 1}\rho_{\lambda, 0}A_{\lambda, 1}A_{\lambda, 2}A_{\lambda, 3}\rho_{\lambda, 0}A_{\lambda, 3}\right)$ and then follow similar steps as~\eqref{eq:thm2_main_step}, now taking the difference with the $\mathcal{F}(\tilde{\rho}_{\lambda, 2}, \rho_{\lambda, 0})$. We get, 
    \begin{equation}
        |{\mathcal{F}}(\tilde{\rho}_{\lambda, 3}, \rho_{\lambda, 0}) - {\mathcal{F}}(\tilde{\rho}_{\lambda, 2}, \rho_{\lambda, 0})| \leq \sqrt{2\delta}.
    \end{equation}
    From this, we get, ${\mathcal{F}}(\tilde{\rho}_{\lambda, 3}, \rho_{\lambda, 0}) \geq 1- 2\sqrt{2\delta}-\delta$. In general, we can repeat this step $t_f - 1$ times, to get
    \begin{equation}
        |{\mathcal{F}}(\tilde{\rho}_{\lambda, t_f}, \rho_{\lambda, 0}) - {\mathcal{F}}(\tilde{\rho}_{\lambda, t_f-1}, \rho_{\lambda, 0})| \leq \sqrt{2\delta}.
    \end{equation}
    So, ${\mathcal{F}}(\tilde{\rho}_{\lambda, t_f}, \rho_{\lambda, 0}) \geq 1 - (t_f-1)\sqrt{2\delta} + {O}(\delta)$. Now, note that the normalised fidelity, 
    \begin{equation}
        \mathcal{F}(\rho_{\lambda, t_f}, \rho_{\lambda, 0}) = \frac{{\mathcal{F}}(\tilde{\rho}_{\lambda, t_f}, \rho_{\lambda, 0})}{\tilde{P}_{\lambda, t_f}} \geq \mathcal{F}(\tilde{\rho}_{\lambda, t_f}, \rho_{\lambda, 0}).
    \end{equation}
    From this, we get
    \begin{equation}
        \begin{split}
            \mathcal{F}({\rho_{\lambda, t_f+1}, \rho \otimes |0\rangle \langle 0|^{\otimes m}}) &\geq \mathcal{F}(\tilde{\rho}_{\lambda, t_f}, \rho_{\lambda, 0}) \\
            & \geq 1 - (t_f - 1)\sqrt{2\delta} + {O}(\delta).
        \end{split}
    \end{equation}

     We now calculate the probability that the algorithm halts after the $t_f+1$ steps. In particular, we want to show that $P_{\lambda, t}$ (the probability of observed outcome) is : (i) high for projection $(I-\Pi)$ if wrong list operator is chosen ($t \leq t_f$) (ii) high for projection on $\Pi$ when correct list operator is chosen $(t = t_f +1)$. To show this, note that the fidelity $\mathcal{F}(\rho_{\lambda, t-1}, \rho_{\lambda, 0}) \geq 1 - (t-2)\sqrt{2\delta} + \Theta(\delta)$ with high probability over $\lambda$. So, we can upper bound the trace distance between these states as,
    \begin{equation}
    \begin{split}
        \text{TD}\left(\rho_{\lambda, t-1}, \rho_{\lambda, 0}\right) &\leq \sqrt{1 - \mathcal{F}(\rho_{\lambda, t-1}, \rho_{\lambda, 0})} \\
        &\leq \sqrt{(t-2)\sqrt{2\delta} - \Theta(\delta)}.
    \end{split}
    \end{equation}
    Now, using the identity $\text{TD}(U\rho U^{\dagger}, U\sigma U^{\dagger}) = \text{TD}(\rho, \sigma)$ for any unitary $U$. Thus, we get
    \begin{equation}
        \begin{split}~\label{eq:TD_1}
            \text{TD}(U_{\lambda}^{\dagger}  E^{s}_{t}\rho_{\lambda, t-1}E^{s}_{t}U_{\lambda}, U_{\lambda}^{\dagger}E^{s}_{t}\rho_{\lambda, 0}E^{s}_{t}U_{\lambda}) 
            \leq \sqrt{(t-2)\sqrt{2\delta} - \Theta(\delta)}.
        \end{split}
    \end{equation}
    Now, by definition of trace distance, for any 2-outcome POVM $\{E_0, E_1\}$, we have,
    \begin{equation}\label{eq:TD_2}
        d_{TV}(p, q) \leq \text{TD}(\rho, \sigma),
    \end{equation}
    where $d_{TV}(p, q)$ is total variational distance between distributions: $p_i = \tr E_i \rho$ and $q_i = \tr E_i \sigma$. We take the POVM as $\{\Pi, (I-\Pi)\}$ with $\Pi = I \otimes |0\rangle \langle 0|^{\otimes m}$. Then from~\eqref{eq:TD_1} and~\eqref{eq:TD_2}, we get that
    \begin{equation}
        \begin{split}
           \left|\tr((I-\Pi)U_{\lambda}^{\dagger}  E^{s}_{t}\rho_{\lambda, t-1}E^{s}_{t}U_{\lambda}) 
           - \tr\left((I-\Pi)U_{\lambda}^{\dagger} E^{s}_{t}\rho_{\lambda, 0}E^{s}_{t}U_{\lambda}\right)\right| 
           &\leq 2\sqrt{(t-2)\sqrt{2\delta} - {O}(\delta)}
        \end{split}
    \end{equation}
    and,
    \begin{equation}
        \begin{split}
           \left|\tr(\Pi U_{\lambda}^{\dagger}  E^{s}_{t}\rho_{\lambda, t-1}E^{s}_{t}U_{\lambda}) - \tr\left(\Pi U_{\lambda}^{\dagger} E^{s}_{t}\rho_{\lambda, 0}E^{s}_{t}U_{\lambda}\right)\right| 
           \leq 2\sqrt{(t-2)\sqrt{2\delta} - {O}(\delta)}. 
        \end{split}
    \end{equation}
    Denote $\beta = 2\sqrt{(t-2)\sqrt{2\delta} - {O}(\delta)}$.
    Thus, distribution over measurement outcomes for POVM $\{\Pi, I- \Pi\}$ for the state $U_{\lambda}^{\dagger} E_t^s \rho_{\lambda, t-1}E_t^s U_{\lambda}$ is $\beta$-close in total variational distance to that for the state $U_{\lambda}^{\dagger} E_t^s \rho_{\lambda, 0}E_t^s U_{\lambda}$
    Now, notice that if $E^{s}_{t}$ above is wrong list error, then $\tr((I-\Pi)U_{\lambda}^{\dagger}  E^{s}_{t}\rho_{\lambda, t-1}E^{s}_{t}U_{\lambda}) = P_{\lambda, t}$ and thus we get $|P_{\lambda, t} - P_{\lambda, 1}| \leq \beta$. Similarly if $E^{s}_{t}$ is correct list error, we are interested in the probability of observing outcome $\Pi$ in the POVM which from above will be $\beta$-close to probability of observing $\Pi$ on the state $U_{\lambda}^{\dagger} E^{s}_{t}\rho_{\lambda, 0}E^{s}_{t}U_{\lambda}$ which is 1. Thus, we get, 
    \begin{equation}
        P_{\lambda, t} \geq 1 - {O}(2^{-m}) - \beta, \hspace{5ex} t \leq t_f
    \end{equation}
    and, 
    \begin{equation}
        P_{\lambda, t_f+1} \geq 1 - \beta
    \end{equation}
    where we just used the expressions for $P_{\lambda, 1}$ from Proposition~\ref{thm:prob_calcs_main_text} and noted that $\tr(\Pi U_{\lambda}^{\dagger} E^{s}_{t_f+1}\rho_{\lambda, 0}E^{s}_{t_f+1}U_{\lambda}) = 1$. \\
    Now, as mentioned earlier, for our algorithm to succeed, for the first $t_f$ steps, the measurement result should be $(I-\Pi)$ and in $t_f +1$ step, the result should be $\Pi$, which happens with probability, 
    \begin{equation}
        \begin{split}
            \Pi_{i=1}^{t_f+1} P_{\lambda, i} &\geq (P_{\lambda, 1} -\beta )^{t_f} (1-\beta) 
            = \Big(1- {O}(2^{-m}) - \beta\Big)^{t_f}(1-\beta). \\
        \end{split}
    \end{equation}
    Thus, we get that with probability $1 - O\left(\frac{2^{-m}}{\delta}\right)$ over the keys $\lambda$, our algorithm on input state of form $EU_{\lambda} \rho \otimes |0\rangle \langle 0| U_{\lambda}^{\dagger} E$, returns the state $\rho_{\lambda, t_f+1}$ which has fidelity $\mathcal{F}(\rho_{\lambda, t_f+1}, \rho \otimes |0\rangle \langle 0|^{\otimes m}) \geq 1 - (t_f - 1)\sqrt{2\delta} + {O}(\delta)$ with success probability $P_{\text{succ}} \geq \Big(1- {O}(2^{-m}) - \beta\Big)^{t_f}(1-\beta) $ with $\beta = 2\sqrt{(t-2)\sqrt{2\delta} - {O}(\delta)}$. Now, for $\delta = \Theta(2^{-m + \text{polylog}(n)})$, $m = \omega(\log n)$ and $t_f + 1 \leq L = {O}(\poly(n))$, we get, with probability $1 - O\left(\frac{2^{-m}}{2^{-m+\text{polylog}(n)}}\right) \sim 1 - \negl(n)$ over keys $\lambda$, 
    \begin{equation}
        \begin{split}
            \mathcal{F}(\rho_{\lambda, t_f+1}, \rho \otimes |0\rangle \langle 0|^{\otimes m}) 
            \geq 1 - {O}(\poly(n)2^{-\omega(\log n) + \text{polylog}(n)})
            \sim 1 - \negl(n)
        \end{split}
    \end{equation}
    and $\beta = 2\sqrt{(t-2)\sqrt{2\delta}} - O\left(\frac{\delta^{3/2}}{\sqrt{t-2}}\right) \sim \negl(n)$
    \begin{equation}
    \begin{split}
        P_{\text{succ}} &\geq 1 - {O}(t_f 2^{-\omega(\log n)}) + \beta t_f - {O}(\beta)\\
        &\sim 1- \negl(n).
    \end{split}
    \end{equation}
\end{proof}

Note that while in the above, we explicitly work with pure states and particular measurement outcomes at each step, it is also easy to see that our algorithm can be thought of as a linear map, by using POVM elements $\{\Pi_0 \otimes |0\rangle, (I-\Pi_0) \otimes |1\rangle\}$ for measurement. To apply operations conditional on measurement, we condition the measurement on the classical outcome, which is stored in a classical register. If the outcome is $0$, we leave the state invariant, while if the outcome is $1$, we apply the inverse of the wrong correction operator and apply the correction operator for the next list element. We iterate over all list elements, i.e. we apply this channel $l \leq L$ times. Then, after $l$ steps, we will have a mixed classical and quantum state with the classical register of size $l$ qubits, measuring on which will give the state which will have very high fidelity with the original state with overwhelming probability as calculated above. Since this linear channel will work for all pure states $E_i |\psi\rangle \langle \psi|E_i^{\dagger}$ with $E_i$ in the list, we highlight that it is also correct when the input state is a mixture over all list errors, matching our original definition of quantum list decoding.
    
\begin{corollary}[Restatement of \Cref{corr:AQECC}]
    
     Let $Q_L$ be the $[[n^{\prime}, n+m]]$ stabilizer code which is also $(\alpha, L)$ list decodable with encoding and decoding operations as $\mathrm{Enc}_S$ and $\mathrm{Dec}_S$ respectively,  and let $\{U_{\lambda}\}_{\lambda}$ be the PRU with exact $2$-design and relative error $\epsilon$ approximate $4$-design property. Then, define the pair of algorithms $\mathrm{Enc}$ and $\mathrm{Dec}$ as $\mathrm{Enc}(k, |\psi\rangle \langle\psi|) = \mathrm{Enc}_S (U_{\lambda} |\psi\rangle \langle\psi| \otimes |0\rangle\langle0|^{\otimes m} U_{\lambda}^{\dagger})$ and decoding operation as $\mathrm{Dec}(k, \rho)  = \mathrm{Dec}_A(k, \mathrm{Dec}_{S}(\rho))$ where $\mathrm{Dec}_{ A}$ is the decoder as in the Algorithm~\ref{alg:unambiguous_qld}, then $(\mathrm{Enc}, \mathrm{Dec})$ is a $(\mathcal{A}, \sqrt{L\sqrt{2\delta}})$ cryptographic approximate quantum error correcting code where $\mathcal{A}$ is $(\alpha,t)$ computationally bounded adversary defined in Def.~\ref{def:advnoise} for $t = O(\poly (n))$, $\delta = \Theta(2^{-m + \poly \log n})$ is the fidelity error as in Proposition~\ref{prop:fidelity_calcs_main_text} and $\epsilon = O(2^{-m})$.
\end{corollary}
\begin{proof}

Firstly, we argue that since $\{U_{\lambda}\}_{\lambda}$ is PRU, no computationally bounded adversary with polynomial queries to it can distinguish it from Haar random unitaries. Consequently, there is no way for a computationally bounded adversary to learn the key used in the protocol. Thus, the error channel used by the adversary will be independent of the key. \\
As the adversary can use only Pauli errors of weight at most $\alpha n$ independent of the key, we can use the algorithm~\ref{alg:unambiguous_qld} as $\text{Dec}_A$ along with the decoder of QLD, $\text{Dec}_S$ as the decoder of the cryptographic approximate quantum error correcting code. Then, we note that from Theorem~\ref{thm:almost_pure_state}, for any input state $|\phi\rangle \langle \phi|$, the fidelity, 
    \begin{equation}
        \mathcal{F}(\text{Dec}_\lambda \circ \mathcal{A} \circ \text{Enc}_\lambda (|\phi\rangle \langle \phi|), |\phi\rangle \langle \phi|) \geq 1 - L\sqrt{2\delta}
    \end{equation}
    with probability $1 - O\left(\frac{\epsilon + 2^{-m}}{\delta}\right)$ over both keys $\lambda$ and randomness in the algorithms for $\delta = \Theta(2^{-m + \poly \log n})$. Note that $\mathcal{A}$ above represents an adversarial channel with bounded computational power and weight $\alpha \cdot n$. Then, denoting the above fidelity as $\mathcal{F}_\lambda$, we have
    \begin{equation}
        \Eset{\lambda}\,\mathcal{F}_\lambda \geq \left(1 - O\left(\frac{\epsilon + 2^{-m}}{\delta}\right)\right) (1 - L\sqrt{2\delta}).
    \end{equation}
    Now, using the Fuchs-van de Graaf relation between trace distance and fidelity, $\text{TD}(\rho, \sigma) \leq \sqrt{1 - \mathcal{F}(\rho, \sigma)}$,
    \begin{equation}
        \Eset{\lambda}\,\text{TD}(\text{Dec}_\lambda \circ \mathcal{A} \circ \text{Enc}_\lambda (|\phi\rangle \langle \phi|), |\phi\rangle \langle \phi|) \leq  \Eset{\lambda}\, \sqrt{1 - \mathcal{F}_\lambda}.
    \end{equation}
    Then, using concavity of the function $f(x) = \sqrt{1- x}$, we get that,
    \begin{equation}
    \begin{split}
        \Eset{\lambda}\,\text{TD}(\text{Dec}_\lambda \circ \mathcal{A} \circ \text{Enc}_\lambda (|\phi\rangle \langle \phi|), |\phi\rangle \langle \phi|) &\leq  \sqrt{1 - \Eset{\lambda}\,\mathcal{F}_\lambda}\\
        & \leq \sqrt{L\sqrt{2\delta} + O\left(\frac{\epsilon + 2^{-m}}{\delta}\right)}.
    \end{split}
    \end{equation}
    Now, it is easy to see that for $\delta = \Theta(2^{-m + \poly \log n})$ and $\epsilon = O(2^{-m})$ gives the upper bound on expected trace distance as $\sqrt{L\sqrt{2\delta}}$. Since it is true for any pure state $|\phi\rangle \langle \phi|$, we get
    \begin{equation}
        \sup_{|\phi\rangle}\Eset{\lambda}\,\text{TD}(\text{Dec}_\lambda \circ \mathcal{A} \circ \text{Enc}_\lambda (|\phi\rangle \langle \phi|), |\phi\rangle \langle \phi|) \leq  \sqrt{L\sqrt{2\delta}}.
    \end{equation}
\end{proof}

\end{document}